\documentclass[a4paper, 12pt]{amsart} 

\usepackage{fullpage}
\usepackage[numbers,sort&compress,merge]{natbib}

\newcommand{\abs}[1]{\left \lvert #1 \right \rvert}

\newcommand\J{\mathrm{J}_{a,b}}
\newcommand\Lag{\mathrm{L}_b}
\newcommand\G{\mathrm{G}}
\newcommand{\av}[1]{\left \langle #1 \right \rangle}

\DeclareMathOperator{\tr}{Tr}

\theoremstyle{plain}
\newtheorem{theorem}{Theorem}[section]
\newtheorem{proposition}[theorem]{Proposition}
\newtheorem{lemma}[theorem]{Lemma}
\newtheorem{corollary}[theorem]{Corollary}
\theoremstyle{definition}
\newtheorem{remark}[theorem]{Remark}

\theoremstyle{definition}

\begin{document}
\title{Moments of the transmission eigenvalues, proper delay times 
   and random matrix theory I} 

\author{F.~Mezzadri}
\address{School of Mathematics, University of Bristol, Bristol BS8
 1TW, UK}
\email{f.mezzadri@bristol.ac.uk}

\author{N.~J.~Simm}
\address{School of Mathematics, University of Bristol, Bristol BS8
 1TW, UK}
\email{n.simm@bristol.ac.uk}

\thanks{Research partially supported by EPSRC, grant no: EP/G019843/1,
  and by the Leverhulme Trust, Research Fellowship no:
  RF/4/RFG/2009/0092.}
\begin{abstract}
 We develop a method to compute the moments of the eigenvalue
 densities of matrices in the Gaussian, Laguerre and Jacobi ensembles
 for all the symmetry classes $\beta \in \{1,2,4\}$ and finite matrix
 dimension $n$. The moments of the Jacobi ensembles have a physical
 interpretation as the moments of the transmission eigenvalues of an
 electron through a quantum dot with chaotic dynamics.  For the
 Laguerre ensemble we also evaluate the finite $n$ \textit{negative}
 moments.  Physically, they correspond to the moments of the proper
 delay times, which are the eigenvalues of the Wigner-Smith matrix.
 Our formulae are well suited to an asymptotic analysis as $n \to
 \infty$.
\end{abstract}
\maketitle
\tableofcontents

\section{Introduction}
\subsection{Background}
Over the past twenty years Random Matrix Theory (RMT) has provided a
powerful tool to investigate quantum properties of electronic
transport through ballistic cavities (quantum
dots)~\cite{Bee93,BM94,JPB94,BB96,Bee97}.

The purpose of this paper is to compute averages of the form
\begin{equation}
\label{integrals}
M_{\mathcal{E}}^{(\beta)}(k,n) = 
\frac{1}{C}\int_{I}\dotsm
\int_{I}\left(\sum_{j=1}^{n}x_{j}^{k}\right)
\prod_{j=1}^{n}
w_{\beta}(x_{j})\prod_{1 \leq j < k \leq n}|x_{k}-x_{j}|^{\beta}dx_{1}
\dotsm dx_{n}
\end{equation}
for finite $n$ and $k$ and for any value of $\beta \in \{1,2,4\}$.
Here $\mathcal{E}$ labels one of the Gaussian ($\mathrm G$), Laguerre
($\mathrm{L}_{b}$) or Jacobi ($\mathrm{J}_{a,b}$) ensembles and the
value of $\beta$ corresponds to ensembles of real symmetric
($\beta=1$), complex hermitian ($\beta=2$) or quaternion self-dual
matrices ($\beta=4$).  The function $w_\beta(x)$ is the weight of the
ensemble:
\begin{equation}
\label{weights}
w_{\beta}(x) = \begin{cases}
 e^{-\beta x^{2}/2}, & I=(-\infty,\infty), \; \;
\text{Gaussian ensembles,} \\
 x^{\beta/2 (b+1)-1}e^{-\beta x/2},& I=[0,\infty),
 \qquad \text{Laguerre ensembles,}  \\
 x^{\beta/2 (b+1)-1}(1-x)^{\beta/2 (a+1)-1}, &
 I=[0,1], \qquad \; \; \text{Jacobi ensembles},
   \end{cases}
\end{equation}
where $C$ is a normalization constant which may vary at each
occurrence.  The averages~\eqref{integrals} for the Jacobi ensembles
correspond to the moments of the \textit{transmission eigenvalues} of
the electric current through a ballistic cavity; the \textit{negative}
moments of the Laguerre ensembles are the moments of the density
of the eigenvalues of the Wigner-Smith time-delay matrix.

The physical dimensions of mesoscopic systems are such that the
quantum nature of the electron becomes important and a classical
treatment of its dynamics is not accurate anymore. Furthermore, at low
temperature and voltage electron-electron interactions can be
neglected; therefore, the electron scatters elastically inside the
cavity, which is attached to two ideal leads connecting two reservoirs
in equilibrium at zero temperature.  If the leads support $m$ and $n$
quantum channels respectively, all the information on the electric
transport is contained in the scattering matrix
\begin{equation}
 \label{eq:scattering_matrix}
S =  \begin{pmatrix}
   r_{m \times m} & t'_{m \times n} \\
   t_{n \times m} & r'_{n \times n}
 \end{pmatrix}.
\end{equation}
The sub-blocks $r_{m \times m}$ and $t_{n \times m}$ are the
reflection and transmission matrices through the left lead, while
$t'_{m \times n}$ and $r'_{n \times n}$ are those through the right
lead.  Without loss of generality, throughout this paper we shall
assume that $m \ge n$. Since the scattering is elastic $S$ is
unitary. This is known as the Landauer-B\"uttiker formalism.

The scattering matrix depends in a complicated way on macroscopic
parameters like the energy of the electron and the shape of the
cavity. If the classical dynamics inside the cavity is chaotic, then
the fundamental assumption is that the electric current displays
universal features; thus, it is natural to model the scattering matrix
$S$ with a random matrix drawn from one of Dyson's circular ensembles:
the Circular Unitary Ensemble (CUE) when $\beta=2$; the Circular
Orthogonal Ensemble (COE) when $\beta=1$; and the Circular Symplectic
Ensemble (CSE) when $\beta=4$.  Let $K$ denote a time reversal
operator. If the dynamics is not time-reversal invariant then
$\beta=2$; if it is time-reversal invariant then $\beta=1$ if $K^2=1$
and $\beta=4$ if $K^2=-1$.

In this paper we give a unified approach to compute the family of
integrals~\eqref{integrals} for all the $\beta \in \{1,2,4\}$ and give
particular emphasis to those connected to statistics of the electric
current.  Our formulae are exact for finite matrix dimension. Since
experiments can now be performed in quantum dots with a number of
channels arbitrarily small~\cite{OSSSHH01}, recently there has been an
increasing interest in computing finite $n$
formulae~\cite{VV08,Nov08,OK08,*OK09,KSS09}. Some of these integrals
have never been computed before, others are already available in the
literature~\cite{HZ86,GJ97,Nov08,VV08,Led09}. In particular, most of
the formulae for $\beta=1$ and $\beta=4$ are original.  In
\S\ref{se:results} we will discuss in detail our results and specify
which of the averages~\eqref{integrals} are already known.

Our formulae have distinctive advantages. Firstly, we can compute
\textit{`negative'} moments for the Laguerre ensemble.  Since the
joint probability density function (\textit{j.p.d.f.}) of the
\textit{inverse delay times} coincides with that of the Laguerre
ensemble~\cite{BFB97}, we obtain the moments of their density.
Secondly, for positive moments the sums in our formulae extend to the
order of the moments $k$ and not to the dimension of the matrices
$n$. The sums that express the negative moments in the Laguerre
ensemble run to $n$, but their limit as $n \to \infty$ can be computed
with little effort.  As a consequence, although still relatively
involved, our expressions are simpler and more manageable than those
in the literature. Furthermore, our formulae provide a bridge
between finite $n$ results and their asymptotics.  Indeed, in the
second part of this work~\cite{MS11b}, we compute the first three
terms of the expansions as $n \to \infty$ of the moments of the
transmission eigenvalues and of the delay times. They agree with those
recently obtained semiclassically~\cite{BHN08,BK10,BK11}.

\subsection{The Transmission Eigenvalues}
The eigenvalues $T_1,\dotsc,T_n$ of the matrix $tt^\dagger$ are the
\textit{transmission eigenvalues}.  The unitarity of $S$ implies that
the $T_1,\dotsc,T_n$ all lie in the interval $[0,1]$.  The
dimensionless conductance at zero temperature is defined by
\begin{equation}
 \label{eq:conductunce_def}
 G := \tr tt^\dagger = \tr t't^{\prime \dagger} = T_1 + \dotsb + T_n.
\end{equation}
Furthermore, if $S$ belongs to one of Dyson's circular
ensembles, then the \textit{j.p.d.f.} of $T_1,\dotsc,T_n$ is
\begin{equation}
 \label{eq:jpdf_trans_eig}
 p^{(\beta)}\left(T_1,\dotsc,T_n\right) = \frac{1}{C
 }\prod_{j=1}^{n}T_{j}^{\alpha} \prod_{1 \leq j < k \leq n}\left
   \lvert T_{k}-T_{j}\right \rvert^{\beta}.
\end{equation}
The parameter $\alpha = \frac{\beta}{2}\left( m-n + 1\right) - 1$
measures the asymmetry of the quantum channels in the
leads. Formula~\eqref{eq:jpdf_trans_eig} was first computed when $m=n$
by Baranger and Mello~\cite{BM94} and by Jalabert \textit{et
 al}~\cite{JPB94}; when $m \neq  n$ it was reported in this form by
Beenakker~\cite{Bee97}, where it was attributed to unpublished work by
Brouwer (1994); its general derivation appeared in the literature for
the first time in an article by Forrester~\cite{For06}.

In a classic article Dyson~\cite{Dys62d} classified complex many-body
systems according to their fundamental symmetries and proved that they
correspond to the random matrix ensembles labelled $\beta\in
\{1,2,4\}$. Zirnbauer~\cite{Zir96} extended Dyson's classification
scheme to Cartan's symmetric spaces and introduced new symmetry
classes in Random Matrix Theory. Zirnbauer also argued that these
non-standard ensembles appear in the stochastic modelling of ballistic
cavities in contact with a superconductor.  Such mesoscopic systems
are called \textit{Andreev quantum dots.}  In his PhD thesis
Due\~{n}ez~\cite{Due01,*Due04} further generalized Zirnbauer's
classification. Furthermore, Altland and Zirnbauer~\cite{AZ96,*AZ97}
divided the symmetries of Andreev quantum dots into four fundamental
classes.  These ensembles are labelled by two integers
$(\beta,\delta)$: as for Dyson's ensembles, $\beta$ takes values in
$\{1,2,4\}$; instead $\delta \in \{-1,1,2\}$.  The four classes are
$(1,-1)$, $(2,-1)$, $(4,2)$ and $(2,1)$; they correspond to different
combinations of time-reversal and spin-rotation symmetries.

Our formalism applies to Andreev quantum dots too.  Indeed, in a
recent paper Dahlhaus \textit{et al}~\cite{DBB10} computed the
\textit{j.p.d.f.} of the transmission eigenvalues.  It is obtained by
deforming the right-hand side in equation~\eqref{eq:jpdf_trans_eig}:
\begin{equation}
 \label{eq:tr_eig_Aqd}
 p^{(\beta,\delta)}(T_1,\dotsc,T_n) = \frac{1}{C
 }\prod_{j=1}^{n}T_{j}^{\alpha}\left(1 - T_j\right)^{\delta/2} 
 \prod_{1 \leq j < k \leq n}\left.
   \lvert T_{k}-T_{j}\right \rvert^{\beta}
\end{equation}
As for the \textit{j.p.d.f.} in~\eqref{eq:jpdf_trans_eig} $\alpha =
\frac{\beta}{2}\left(m-n +1\right) - 1$.

Equations~\eqref{eq:jpdf_trans_eig} and~\eqref{eq:tr_eig_Aqd}
are  particular cases of the \textit{j.p.d.f.} of the eigenvalues of
matrices in the Jacobi ensembles, namely
\begin{equation}
 \label{eq:jpdf_Jacobi}
 p^{(\beta)}_{\mathrm{J}_{a,b}}(x_1,\dotsc,x_n) =\frac{1}{C} \prod_{j=1}^n  
 x_j^{\beta/2 (b+1)-1}(1-x_j)^{\beta/2 (a+1)-1}\prod_{1 \le j <k \le n}
 \left \lvert x_k - x_j\right \rvert^\beta,
\end{equation}
for $0 \le x_j \le 1$, $(j=1,\dotsc,n)$. We
recover~\eqref{eq:tr_eig_Aqd} by setting
\begin{equation}
 \label{eq:a_an_b}
 a =\frac{2}{\beta}\left(1 + \frac{\delta}{2}\right) - 1
 \quad  \text{and} \quad b = m-n.
\end{equation}

The moments of the density of the transmission eigenvalues are defined by
\begin{equation}
 \label{trans_mom}
\av{\mathcal{T}_{k,n,m}^{\left(\beta,\delta\right)}} :=\av{\tr \bigl 
[\left(tt^\dagger\right)^k\bigr]} =   M^{(\beta)}_{\mathrm{J}_{a,b}}(k,n),
\end{equation}
where $a$ and $b$ are given in equation~\eqref{eq:a_an_b}.  From a
physical point of view, they are important because they are connected
to the cumulants $\av{\av{\kappa_j}}$ of the charge transmitted over a
finite interval of time by the generating function~\cite{LLY95} 
\begin{equation}
 \label{eq:gen_fun}
 \sum_{j=1}^{\infty}\frac{x^{j}}{j!}\av{\av{\kappa_{j}}}
 = -\sum_{k=1}^{\infty}\frac{(-1)^{k}}{k} \av{\mathcal{T}_{k}} (e^{x}-1)^{k}.
\end{equation}
(See also~\cite{BB00}, appendix A.) For simplicity in this
formula we have omitted the dependence on $(\beta,\delta)$ and on the
numbers of quantum channels $m$ and $n$.  The charge cumulants can be
directly accessed in experiments~\cite{BGSLR05}. Using the results in
\S\ref{se:results} and the generating function in~\eqref{eq:gen_fun}
we can compute the cumulants to any given order. For example, the
variance and skewness are given by
\begin{subequations}
\label{eq:cumulants_12}
\begin{align}
\langle \langle \kappa_{2} \rangle \rangle & =
\frac{nm(\frac{2+\delta}{\beta}
-1+n)(\frac{2+\delta}{\beta}-1+m)}{(\frac{4+\delta}{\beta}-1+n+m)
(\frac{2+\delta}{\beta}-2+n+m)(\frac{2+\delta}{\beta}-1+n+m)}\\
\frac{\langle \langle \kappa_{3} \rangle \rangle}{\langle \langle
 \kappa_{2} \rangle \rangle}&  = -\frac{(n-m-\frac{2+\delta}{\beta}+1)
(n-m+\frac{2+\delta}{\beta}-1)}%
{(n+m+\frac{2+\delta}{\beta}-3)(n+m+\frac{6+\delta}{\beta}-1)}.
\end{align}
\end{subequations}
The special case $\delta=0$ of equations (\ref{eq:cumulants_12}) were
computed by Savin \textit{et al}~\cite{SSW08} (see also~\cite{BSB01}).

\subsection{The Wigner-Smith Matrix}
The Wigner-Smith time-delay matrix is defined as
\begin{equation}
 \label{eq:wi_del_time_mat}
Q   = -i\hbar S^{-1}\frac{\partial S}{\partial E}.
\end{equation}
The individual eigenvalues $\tau_1, \dotsc,\tau_n$ of $Q$ are
called \textit{proper delay times}, and their average
\begin{equation}
 \label{eq:wig_del_time}
 \tau_{\mathrm{W}} = \frac{1}{n} \tr Q
\end{equation}
is referred to as \textit{Wigner delay time.}  Here $n$ is the total
number of quantum channels in the leads. The Wigner delay time
measures the extra time an electron spends in the cavity as a result
of being scattered.  If $S$ belongs to one of the circular ensembles,
then it was shown by Brouwer \textit{et al}~\cite{BFB97} that the
\textit{j.p.d.f.} of the inverses $\gamma_j = \tau_j^{-1}$
($j=1,\dotsc,n$) of the proper delay times is
\begin{equation}
\label{jpdfdelaytime}
P_{\beta}(\gamma_{1},\ldots,\gamma_{n}) = \frac{1}{C} 
\prod_{j=1}^{n}\gamma_{j}^{ n\beta/2}e^{-\beta \tau_{\mathrm{H}} 
\gamma_{j}/2} \prod_{1 \leq j < k \leq n} \abs{\gamma_{k}-\gamma_{j}}^{\beta},
\end{equation}
where $\tau_{\mathrm{H}}$ is the Heisenberg time.  In our context
$\tau_{\mathrm{H}}=n$.  In a sequence of papers Savin and
collaborators~\cite{SFS01,SSS01,SS03} computed the probability
distribution function of the proper delay times.

The moments of the density of the proper delay times are defined by
\begin{equation}
   \label{eq:mom_time_del}
   \left \langle \mathcal{D}^{(\beta)}_{k,n}\right \rangle = 
  \frac{1}{n}\left \langle \tr Q^k\right \rangle =
   n^{k-1}M^{\left(\beta\right)}_{\mathrm{L}_b}(-k,n), \quad k < n\beta/2 + 1, 
\end{equation}
where in this case $b = n-1 + 2/\beta$, and the right-hand side of
(\ref{eq:mom_time_del}) denotes negative integer moments of the
Laguerre ensemble. 

The outline of the paper is the following: in \S2 we present our
main results; \S3 is devoted to ensembles with unitary symmetry;
in \S4 we discuss the approach underlying the computations of the
moments for ensembles with symplectic and orthogonal symmetries;
\S5 and \S6 contain the proofs of the formulae for ensembles with
symplectic and orthogonal symmetries, respectively.

In the final stages of the preparation of this article, and after the
results in this paper had been presented at two
workshops,\footnote{Random Matrix Theory and Its Applications I, MSRI,
  13--17 September 2010, Berkeley, USA; VI Brunel Workshop on Random
  Matrix Theory, 17--18 December 2010, Brunel University, UK.} we
received a preprint by Livan and Vivo~\cite{LV11}, in which some of
our formulae were derived with a different method and approach.  Their
expressions are different, but equivalent to ours.

\section*{Acknowledgements}
We would like to express our gratitude to Peter Forrester, Yan
Fyodorov, Jonathan Keating, Jack Kuipers, Marcel Novaes and Dmitry Savin for
stimulating and helpful discussions.  

\section{\label{se:results}Statement of Results}
\subsection{The Moments of Transmission Eigenvalues and of the Proper
  Delay Times}
Since the transmission eigenvalues are distributed with the
\textit{j.p.d.f.} of the Jacobi ensemble, their moments are the
moments of the eigenvalue density of this ensemble for finite matrix
dimension. A 4-\emph{th} order recurrence relation for the exact
moments at $\beta=2$ were first reported Ledoux~\cite{Led04}. Explicit
formulae were then obtained by Novaes~\cite{Nov08} and by Vivo and
Vivo~\cite{VV08}. Bai~\textit{et al}~\cite{BYK87} computed the leading
order term of the asymptotic expansion for $\beta=1$.

From equation~\eqref{jpdfdelaytime}, computing the moments of the
proper delay times is tantamount to calculating the negative moments
of the eigenvalue density of the Laguerre ensemble for finite
$n$. These negative moments have never been determined before, though
positive integer moments were calculated by Hanlon \emph{et
  al}~\cite{HSS92} and Haagerup and Thorbj{\o}rnsen~\cite{HT03}.

It is worth reminding the reader that the moments of the proper delay
times exist only for
\begin{equation*}
 k < \frac{n\beta}{2} + 1,
\end{equation*}
because for larger $k$ the integral $M_{\mathrm{L}_b}(-k,n)$, with $b
= n -1 + 2/\beta$, diverges. 

The general formulae for moments of the Jacobi and Laguerre ensembles
are reported in the main sections of the paper. Throughout, the
notation $(n)_{(k)}= \Gamma(n + k)/\Gamma(n)$ refers to the Pochhammer
symbol; the binomial coefficient can take arbitrary complex
arguments, \textit{i.e.}
\begin{equation}
 \label{eq:gen_bin}
 \binom{k}{j} = 
 \frac{\Gamma(k+1)}{\Gamma(k - j + 1)\Gamma(j+1)}.
\end{equation}
and is defined for negative integers by the limiting form
\begin{equation}
\binom{-k}{j} = (-1)^{j}\binom{k+j-1}{k-1}.
\end{equation}

\subsubsection{Broken Time Reversal ($\beta=2$)}

\begin{theorem}
\label{the:unitary_jl}
The moments of the transmission eigenvalues and of the proper delay
times for $\beta=2$ are
\begin{equation}
\label{utransmissionmom}
\left \langle \mathcal{T}_{k,n,m}^{(2,\delta)}\right  \rangle 
=  \frac{nm}{\delta/2 + n + m}-\sum_{j=1}^{k-1}\frac{1}{j}\sum_{i=1}^{\min(j,n)}
\binom{j}{i}\binom{j}{i-1}\mathcal{U}^{m,n,\delta}_{i,j}
\end{equation}
and
\begin{equation}
\label{eq:finite_n_mom_del}
\left \langle \mathcal{D}_{k,n}^{(2)} \right \rangle 
= \frac{n^{k-1}}{k}\sum_{j=0}^{n-1}\binom{k+j-1}{k-1}\binom{k+j}{k-1}
\frac{(2n)_{(-k -j)}}{(n+1)_{(-j-1)}},
\end{equation}
where
\begin{equation*}
\mathcal{U}^{m,n,\delta}_{i,j} =
\frac{\bigl(\delta/2 + m + n- 2i+j+1\bigr)\bigl(\delta/2+
 m\bigr)_{(j-i+1)}\bigl(m\bigr)_{(j-i+1)}}%
{\bigl(\delta/2 + m + n-i\bigr)_{(j+2)}\bigl(\delta/2 +
 m+n-i+1\bigr)_{(j)}\bigl(\delta/2 + n+1\bigr)_{(-i)}\bigl(n+1\bigr)_{(-i)}}.
\end{equation*}
\end{theorem}
\begin{remark}
  The moments for the Laguerre Unitary Ensemble can be defined even when
  $k$ is complex and have the following particularly simple
  expression:
\begin{equation}
\label{eq:finite_n_mom_laguerre}
M^{(2)}_{\mathrm{L}_{b}}(k,n) = \frac{1}{k}
\sum_{j=0}^{n}\binom{k}{j}\binom{k}{j-1}\frac{(b+n)_{(k-j+1)}}{(1+n)_{(-j)}},
\end{equation}
of which (\ref{eq:finite_n_mom_del}) is a special case. If $k$ is a
positive integer, the sum in~\eqref{eq:finite_n_mom_laguerre} consists
of at most $k$ terms.
\end{remark}
\begin{remark}
\label{narayanaremark}
The coefficients
\begin{equation}
\label{narayanacoefficient}
N(k,j) = \frac{1}{k}\binom{k}{j}\binom{k}{j-1}
\end{equation}
in formulae~\eqref{utransmissionmom} and~\eqref{eq:finite_n_mom_laguerre} appear frequently in
enumerative combinatorics, where they are called
\textit{Narayana numbers}.
\end{remark}

\subsubsection{Conserved Time Reversal with $K^{2}=-1$ ($\beta=4$)}
\begin{theorem}
   The moments of the transmission
  eigenvalues are
\begin{equation}
\label{eq:mom_b4_del}
\left \langle \mathcal{T}_{k,n,m}^{(4,\delta)} \right \rangle = 
\frac{1}{2}\left \langle \mathcal{T}_{k,2n,2m}^{(2,\delta-2)} 
\right \rangle-\sum_{j=1}^{\min \left (\lfloor n \rfloor,
\left \lfloor k/2 \right \rfloor\right)}
\sum_{i=0}^{\min\left(k-2j,2n-2j\right)}\binom{k}{i}\binom{k}{i+2j}
\mathcal{S}_{i,j}^{\delta}(k,m,n)
\end{equation}
The coefficient $\mathcal{S}_{i,j}^{\delta}(k,m,n)$ is
\begin{equation*}
\begin{split}
 \mathcal{S}^\delta_{i,j}(k,m,n) &= 2^{4j-3} \frac{\bigl(\delta/2 +
   2n-i-2j\bigr)_{(i)}\bigl(2m\bigr)_{(k-i-2j+1)} \bigl(\delta/2 +2m
   -1\bigr)_{(k-i-2j+1)}}%
 {\bigl(\delta/4+
   n+1/2\bigr)_{(-j)}\bigl(m\bigr)_{(1-j)}\bigl(\delta/4+
   m-1/2\bigr)_{(1-j)}
   \bigl(2n-2j+1\bigr)_{(-i)} \bigl(n+1\bigr)_{(-j)}}\\
 &\quad \times \frac{\bigl(\delta/2 + 2m+2n-4j\bigr)\bigl(\delta/2 +
   2m+2n-2i-4j+k\bigr)}{\bigl(\delta/2 + 2m+2n-i-2j\bigr)_{(1+k)}
  \bigl(\delta/2 + 2m+2n-i-4j\bigr)_{(1+k)}}.
\end{split}
\end{equation*}
Furthermore, the moments of the proper delay times for $\beta=4$ are
given by
\begin{equation}
 \label{eq:mom_b4_dt}
 \left \langle \mathcal{D}^{(4)}_{k,n}\right \rangle = n^{k-1}
M_{\mathrm{L}_{n-1/2}}^{(4)}(-k,n),
\end{equation}
where the moments of the Laguerre Symplectic Ensemble  are
\begin{equation}
\label{eq:moments_timedelay_4}
\begin{split}
M_{\mathrm{L}_{b}}^{(4)}(k,n) & =
2^{-k-1}M_{\mathrm{L}_{2b}}^{(2)}(k,2n) \\ 
& \quad - \sum_{j=1}^{\lfloor n \rfloor}\sum_{i=0}^{2n-2j}
\binom{k}{i}\binom{k}{i+2j}
\frac{(2b+2n)_{(k-i-2j+1)}(2n-i-2j+1)_{(i)}}{2^{k-2j+2}(n+1)_{(-j)}(b+n)_{(1-j)}}.
\end{split}
\end{equation}
The symbol $\lfloor \cdot \rfloor$ denotes the integer part.
\end{theorem}

\begin{remark}
  The order of the moments $k$ in equations~\eqref{eq:mom_b4_del} and
  \eqref{eq:mom_b4_dt} is a positive integer.  However,
  \eqref{eq:moments_timedelay_4} holds even when $k$ is complex.  As
  for $\beta=2$, if $k$ is positive the sum in
  equation~\eqref{eq:moments_timedelay_4} contains only $k$ terms.
\end{remark}

\begin{remark}
\label{re:half_int}
Although $n$ is typically an integer, as it denotes the dimension of a
matrix, the expressions on the right-hand sides of
equations~\eqref{eq:mom_b4_del} and~\eqref{eq:moments_timedelay_4} are
well defined for any half-integer $n$.  It is useful to generalize it,
because the evaluation of the moments for $\beta=1$ requires moments
for $\beta=4$ computed at half-integer $n$.
\end{remark}

\subsubsection{Conserved Time Reversal with $K^{2}=1$ ($\beta=1$)}
For simplicity we assume that the outgoing lead supports an even
number of open channels. 
\begin{theorem}
\label{th:beta1moments}
The moments of the transmission eigenvalues are
\begin{equation}
\label{eq:mom_b1_del}
 \left \langle \mathcal{T}_{k,m,n}^{(1,\delta)} \right \rangle =
 2\left \langle
   \mathcal{T}_{k,(n-1)/2,(m-1)/2}^{(4,2\delta+4)}\right
 \rangle+\sum_{j=0}^{\min(n/2-1,k)}\binom{2k}{2j}I^{(1,\delta)}_j(k,m,n)
 + \phi^{\mathrm J}_{k,n},
\end{equation}
where
\begin{subequations}
\begin{align}
I^{(1,\delta)}_j(k,m,n) &= 4^{k}
\frac{(\delta+m+n-4j+2k)(\frac{1}{2}(\delta+m+1))_{(k-j)}(\frac{1}{2}m)_{(k-j)}}
{(\delta+m+n-2j)_{(2k+1)}(\frac{1}{2}(\delta+n+2))_{(-j)}
(\frac{1}{2}(1+n))_{(-j)}} \\
\intertext{and}
\phi^{\mathrm J}_{k,n} &= \sum_{j=1}^{k}\frac{2^{\delta+2}
\Gamma(\frac{1}{2}(\delta+m-n+2j+1))\Gamma(\frac{1}{2}(\delta+m+2))}%
{\Gamma(\delta+1+m+j+k)\Gamma(j+k+1-n)\Gamma(\frac{1}{2}(m-n+1+2j))} 
\label{phistatementjacobi}\\
&\quad \times \frac{\Gamma(\frac{1}{2}(\delta+n+2))
\Gamma(m-n+k+j)\Gamma(j+k)}{\Gamma(\frac{m}{2})
\Gamma(\frac{\delta}{2}+1)\Gamma(\frac{n}{2})}. \notag
\end{align}
\end{subequations}
The moments of the proper delay times are
\begin{equation}
\begin{split}
\left \langle \mathcal{D}^{(1)}_{k,n}\right \rangle &=
n^{k-1}2^{1-k}M_{\mathrm{L}_{(n+1)/2}}^{(4)}(-k,(n-1)/2)\\ 
& \quad +\left(\frac{n}{2}\right)^{k-1}\sum_{j=0}^{n/2-1}\binom{2k+2j-1}{2j}\frac{(n+1/2)_{(-j-k)}}%
{2\bigl(\frac{1}{2}(1+n)\bigr)_{(-j)}} + n^{k-1}\phi^{\mathrm L}_{-k,n},
\end{split}
\end{equation}
where
\begin{equation}
\label{eq:phistatementdelay}
\phi^{\mathrm L}_{-k,n} =
\frac{\Gamma(n)}{\Gamma(n/2)\Gamma(2n)}\sum_{j=0}^{k-1}\frac{\Gamma(k+j+n)\Gamma(1+n-k-j)}{\Gamma(n/2+1-j)\Gamma(k+j+1)}2^{j}.
\end{equation}
\end{theorem}
\begin{remark}
  Due to the term $\Gamma(j+k+1-n)$ in the denominator of
  (\ref{phistatementjacobi}), the $\phi^{\mathrm J}_{k,n}$'s are
  identically zero for any $n>2k$. By Stirling's formula, the
  $\phi^{\mathrm L}_{-k,n}$'s in (\ref{eq:phistatementdelay}) decay
  exponentially fast as $n \to \infty$.  Therefore, neither of these
  terms contribute to the asymptotics of the moments as $n \to \infty$
  at any finite algebraic order.
\end{remark}

It is straightforward to compute the limit as $n \to \infty$ of the
formulae in this section. They differ fundamentally from most of the
known exact results in the literature, whose asymptotic limit often
involves many cancellations, which means that even the leading order
term may be out of reach. This difficulty is discussed in some detail
by Krattenthaler~\cite{Kra10}, where a solution is presented for
$\beta=2$ (see also~\cite{CDLV10}).


Indeed, it is a simple exercise using our exact results to show that
\begin{equation}
\label{marcelasympt}
\lim_{n,m \to \infty}\frac{1}{n}\left\langle
 \mathcal{T}_{k,m,n}^{(\beta,\delta)} \right
\rangle = \left(1+\frac{m}{n}\right)\sum_{j=0}^{k-1}\binom{k-1}{j}C_{j}(-1)^{j}\xi^{j+1},
\end{equation}
where $\xi$ is the variable $\xi = \frac{nm}{(n+m)^{2}}$, which
remains finite as $n,m\to \infty$, and $C_{j} =
\frac{1}{j+1}\binom{2j}{j}$ is the $j$-\textit{th} Catalan number.
This formula agrees with the semiclassical computation of Berkolaiko
\textit{et al}~\cite{BHN08}  Furthermore, for the proper delay times
we have
\begin{equation}
\label{schroeder}
\lim_{n \to \infty}\left \langle \mathcal{D}^{(\beta)}_{k,n}\right
\rangle = \frac{1}{k}\sum_{j=0}^{k}\binom{k}{j}\binom{k}{j-1}2^{j},
\end{equation}
which is the $k$-\textit{th} Schr\"oder number (note the appearance of
the Narayana numbers (\ref{narayanacoefficient})). This limit was
computed semiclassically by Berkolaiko and Kuipers~\cite{BK10}, and
can also be obtained from the Mar\v{c}enko-Pastur
distribution~\cite{MP67} (see, \textit{e.g.},~\cite{BFB97,BK10}). It
is a simple consequence of~\eqref{eq:finite_n_mom_del} too.

Equation (\ref{marcelasympt}) was first computed using RMT by
Novaes~\cite{Nov07} (see also~\cite{BYK87}), while (\ref{schroeder})
and (\ref{marcelasympt}) were recently rederived through combinatorial
techniques~\cite{Nov10}. Our exact results allow a simple derivation
of these facts, while also consenting the investigation of
$\beta$-dependent subleading corrections. We address these issues more
thoroughly in the second part of this work~\cite{MS11b}, where we show
that the first two subleading terms in the asymptotic expansions of
the previous theorems agree with those obtained semiclassically by
Berkolaiko and Kuipers~\cite{BK11}.

\subsection{The Gaussian Ensembles}
Our techniques apply equally well to the Gaussian ensembles.
Recursion formulae for the finite $n$ moments of the density of the
eigenvalues were derived by Harer and Zagier~\cite{HZ86} for the
Gaussian Unitary Ensemble (GUE), while Goulden and Jackson~\cite{GJ97}
derived explicit formulae for both the Gaussian Orthogonal Ensemble
(GOE) and the GUE, while the GUE moment generating function was
computed by Haagerup and Thorbj{\o}rnsen~\cite{HT03}. More recently,
recursion formulae were obtained by Ledoux~\cite{Led09} for the GOE
and GSE.
\begin{theorem}
\label{the:gbetae}
The moments of the eigenvalue density for the {\rm GUE} are 
\begin{equation}
\label{gueresult1}
M_{\mathrm{G}}^{(2)}(2k,n) = \frac{2^{n}\Gamma(n/2+1)\Gamma(n/2)}%
{\sqrt{\pi}(2k+1)\Gamma(n)}\sum_{j=0}^{\min(n/2-1,k)}\binom{k}{j}%
\binom{k+1}{j+1}(n/2-j)_{(k+1/2)}
\end{equation}
for even $n$, and
\begin{equation}
\label{gueresult2}
M_{\mathrm{G}}^{(2)}(2k,n) =
\frac{2^{n}\Gamma((n+1)/2)^{2}}{\sqrt{\pi}(2k+1)
\Gamma(n)}\sum_{j=0}^{\min((n-1)/2,k)}
\binom{k}{j}\binom{k+1}{j}((n+1)/2-j)_{(k+1/2)}
\end{equation}
for odd $n$. For the {\rm GSE} we have
\begin{equation}
\begin{split}
M_{\mathrm{G}}^{(4)}(2k,n) & = 2^{-k-1}M_{\mathrm{G}}^{(2)}
 (2k,2n)\\
 &\quad -\frac{\Gamma(n+1)\Gamma(n)}{2^{k}\sqrt{\pi}\Gamma(2n)4^{1-n}}
 \sum_{j=1}^{\min(n,k)}\sum_{i=0}^{\min(n-j,k-j)}
\binom{k}{i}\binom{k}{i+j}(n-i-j+1)_{(k-1/2)}
\end{split}
\end{equation}
Let $n$ be even. Then, the moments for the {\rm GOE} are 
\begin{equation*}
\begin{split}
M_{\mathrm{G}}^{(1)}(2k,n)& =
M_{\mathrm{G}}^{(2)}(2k,n-1)\\
& \quad -\sum_{j=1}^{\min(\frac{n}{2}-1,k)}\sum_{i=0}^{\min(k,\frac{n}{2}-1-j)}\binom{k}{i}\binom{k}{i+j}
\frac{(\frac{n}{2}-i-j)_{(k+1/2)}}{(\frac{n}{2}-j)_{(1/2)}}
+\phi^{\mathrm{G}}_{k,n}.
 \end{split}
\end{equation*}
For $n \le 2k$ the quantity $\phi^{\mathrm{G}}_{k,n}$ is given by
\begin{equation*}
\begin{split}
 \phi^{\mathrm{G}}_{k,n} & = 2^{n/2-k}\frac{(2k)!}{\Gamma(n/2)}
\sum_{j=0}^{k-n/2}\sum_{i=0}^{n/2-1}\frac{\binom{n-1}{2i}
\frac{2^{-j-2i}(-1)^{j}}{(2j+2i+1)j!}}{(k-n/2-j)!}\\
& \quad +\frac{(2k)!}{\Gamma(n/2)}\sum_{j=0}^{n/2-1}
\sum_{i=0}^{j}\frac{(n/2-i-1)!\binom{n-1}{n-2i-1}}{(j-i)!(k-j)!4^{k-i}}.
\end{split}
\end{equation*}
If $n>2k$ we have
\begin{equation}
\label{beta1gauss}
\phi^{\mathrm{G}}_{k,n} = 
(2k)!\sum_{j=0}^{k}\frac{(n/2+1/2-j)_{(j)}}{2^{k-3j}(2j)!(k-j)!}.
\end{equation}
\end{theorem}

\section{\label{sec:unitary_ens}Unitary Ensembles}
We shall now compute the moments of the eigenvalues densities for the
Jacobi, Laguerre and Gaussian ensembles when $\beta=2$.  For brevity
we shall refer to these ensembles with the usual notation JUE, LUE and
GUE. Except for the GUE, our expressions are valid for complex
$k$. Theorem~\ref{the:unitary_jl} and equations~\eqref{gueresult1}
and~\eqref{gueresult2} of theorem~\ref{the:gbetae} are corollaries of
the results of this section.

For all the ensembles and symmetry classes that we consider the 
\textit{j.p.d.f.} of the eigenvalues has the form
\begin{equation}
 \label{eq:gen_jpdf}
 p^{(\beta)}_{\mathcal{E}}(x_1,\dotsc,x_n) =
 \frac{1}{C}\prod_{j=1}^n w_\beta(x_j)
 \prod_{1 \le j < k \le n}\abs{x_k - x_j}^\beta.
\end{equation}
The marginal probabilities are obtained by subsequent integrations of
the right-hand side of~\eqref{eq:gen_jpdf}; furthermore, since it is
invariant under permutations of its arguments, it is irrelevant which
variables are integrated over. Therefore, the probability density of
the eigenvalues is obtained by integrating out all but one variable.
It follows that
\begin{equation}
 \label{eq:traces_gen}
  \begin{split}
 \langle \tr X^k \rangle & = \int_I \dotsi
 \int_I(x_1^k + \dotsb + x^k_n)p^{(\beta)}_{\mathcal{E}}
  (x_1,\dotsc,x_n)dx_1 \dotsm dx_n \\
   & = \int_Ix^k\rho_\beta(x)dx, 
 \end{split}
\end{equation}
where $\rho_\beta(x)$ is the eigenvalue density normalized to $n$ and
$I$ is the support of $w_\beta(x)$.  

We develop effective techniques to compute the integral in the
right-hand side of equation~\eqref{eq:traces_gen} using ideas first
introduced by Haagerup and Thorbj{\o}rnsen~\cite{HT03} for $\beta =2$
and by Adler \textit{et al}~\cite{AFNvM00} for $\beta=1,4$.

When $\beta=2$ the density of the eigenvalues takes a particularly
simple form (see, \textit{e.g.},~\cite{For10}, \S5.1)
\begin{equation}
\label{unitarymed}
\rho_{2}(x) = \left \langle \sum_{j=1}^n \delta(x - x_j)\right
\rangle=  w_{2}(x)\sum_{j=0}^{n-1}\frac{P_{j}(x)^{2}}{h_{j}},
\end{equation}
where the $P_{j}(x)$'s are orthogonal polynomials associated with the
weight $w_2(x)$ and $j=0,1,\dotsc$ denotes their degree.  In other
words, we have 
\begin{equation} 
\label{orthogonality}
\int_I w_{2}(x)P_{j}(x)P_{k}(x)dx = h_{j}\delta_{j k}, \quad
 j,k=0,1,\dotsc .
\end{equation}
The system of orthogonal polynomials $\{P_j(x)\}_{j=0}^\infty$ is
unique up to multiplicative constants $k_j$, which we can take to be
the coefficient of the monomial of highest degree. Orthogonal
polynomials satisfy a recurrence relation of the form
\begin{equation}
\label{3termrec}
P_{j+1}(x) = (\alpha_{j}+x\beta_{j})P_{j}(x)-\gamma_{j}P_{j-1}(x), \quad j=0,1, \dotsc
\end{equation}
where for convention $P_{-1}(x)=0$.  For the classical orthogonal
polynomials the constants $h_j$, $k_j$, $\alpha_j$, $\beta_{j}$ and
$\gamma_{j}$ are tabulated in many books on special functions (see,
\textit{e.g.},~\cite{AS72}). A consequence of~\eqref{3termrec} is
\begin{equation}
\label{christoff}
\rho_{2}(x) = w_2(x)\,\frac{k_{n-1}}%
{k_n h_{n-1}}\left(P'_n(x)P_{n-1}(x)-P_{n}(x)P_{n-1}'(x)\right),
\end{equation}
which is a limiting case of the Christoffel-Darboux formula. (For the
proofs of formulae~\eqref{3termrec} and \eqref{christoff} see,
\textit{e.g.},~\cite{Sze39}, \S3.2). 

In the rest of this article we shall assume that $k_j=1$, for
$j=0,1,\dotsc$.  In other words, we only consider \textit{monic}
orthogonal polynomials.  In order to distinguish them from the way the
classical polynomials are conventionally defined in the literature, we
shall use the notation $\mathcal{H}_{n}(x)$, $\mathcal{L}^{b}_{n}(x)$
and $\mathcal{P}^{a,b}_{n}(x)$ for the Hermite, Laguerre and Jacobi
polynomials respectively.   We shall denote the generic monic
polynomial by $p(x)$.  We tabulate the orthogonality constants
$h_j$ for the monic classical polynomials in appendix A.
We shall also need the following differential equations
(see~\cite{AS72}, \S22.6)
\begin{equation}
\label{odes}
\begin{cases}
\mathcal{H}_{j}''(x)-2x\mathcal{H}_{j}'(x)+2j\mathcal{H}_{j}(x)=0, &\\
x\mathcal{L}_j(x)''+(b+1-x)\mathcal{L}_{j}^{b}(x)'+j\mathcal{L}_{j}^{b}(x)=0,&
\\
x(1-x)\mathcal{P}_{j}^{a,b}(x)''+
(b+1-(a+b+2)x)\mathcal{P}_{j}^{a,b}(x)'& \\
+ j(a+b+j+1)\mathcal{P}_{j}^{a,b}(x)=0. &  \end{cases}
\end{equation}

Haagerup and Thorbj{\o}rnsen~\cite{HT03} computed the moment
generating function
\begin{equation}
\label{momentgenerating}
M(t) = \int_{I}\rho_{2}(x)e^{-tx}dx
\end{equation}
in terms of hypergeometric functions for the GUE and LUE.  They combined the differential equations~\eqref{odes} with~\eqref{christoff} to obtain 
\begin{equation}
\label{differentialidentities}
\frac{d}{dx}\left(f(x)\rho_{2}(x)\right) =
\begin{cases} 
-D_n^{\mathrm H}e^{-x^{2}}\mathcal{H}_{n}(x)\mathcal{H}_{n-1}(x), & \text{Hermite},\\ 
-D_n^{\mathrm L}x^{b}e^{-x}\mathcal{L}^{b}_{n}(x)\mathcal{L}^{b}_{n-1}(x), & \text{Laguerre},
\end{cases}
\end{equation}
where 
\begin{equation}
D_n^{\mathrm H} = \frac{2^{n}}{\sqrt{\pi}\Gamma(n)} \quad \text{and} \quad
D_n^{\mathrm L} = (\Gamma(b+n)\Gamma(n))^{-1}.
\end{equation}
Furthermore, $f(x)=1$ for the Hermite polynomials, while $f(x)=x$ for the Laguerre
ones. We shall use similar ideas to compute the
moments~\eqref{trans_mom} for $\beta=2$.

First we need the analogue of the
identities~\eqref{differentialidentities} for the Jacobi polynomials.

\begin{lemma}
\label{le:jacobi_den}
Let $\rho_{2}(x)$ be the mean eigenvalue density for 
the {\rm JUE}. We have the following differential identity:
\begin{equation}
\label{jacdifferentialidentity}
\frac{d}{dx}\left(x(1-x)\rho_{2}(x)\right) = -D_n^{\mathrm J}
x^{b}(1-x)^{a}\mathcal{P}^{a,b}_{n}(x)\mathcal{P}^{a,b}_{n-1}(x),
\end{equation}
where
\begin{equation}
\label{jacobiceejay}
D_n^{\mathrm{J}} = \frac{\Gamma(a+b+2n+1)\Gamma(a+b+2n-1)}{\Gamma(a+n)\Gamma
(b+n)\Gamma(a+b+n)\Gamma(n)}.
\end{equation}
\end{lemma}
\begin{proof}
 The normalization coefficient $h_{n-1}$ associated with the
 polynomials $\mathcal{P}^{a,b}_{n-1}(x)$ is
\begin{equation}
 \label{jacobipolynormalisation}
 h_{n-1} = \frac{\Gamma(a+n)\Gamma(b+n)\Gamma(n)\Gamma(a+b+n)}%
 {\Gamma(a+b+2n)\Gamma(a+b+2n-1)}.
\end{equation}
Inserting (\ref{jacobipolynormalisation}) into the representation 
(\ref{christoff}) and using the differential equation in~\eqref{odes}
we obtain 
\begin{equation}
\label{jacproofeqn1}
x(1-x)\left(\frac{\rho_{2}(x)}{w_{2}(x)}\right)'
+(1+b-(a+b+2)x)\left(\frac{\rho_{2}(x)}{w_{2}(x)}\right) = -D_n^{\mathrm
 J}\mathcal{P}^{a,b}_{n}(x)\mathcal{P}^{a,b}_{n-1}(x),
\end{equation}
where $D_n^{\mathrm J}$ is given in (\ref{jacobiceejay}).
Finally, since the weight associated with the Jacobi polynomials is
$w_{2}(x) = x^{b}(1-x)^{a}$ we arrive at
\begin{equation*}
\begin{split}
\frac{d}{dx}\left(x(1-x)\rho_{2}(x)\right) &= \frac{d}{dx}
\left(x^{b+1}(1-x)^{a+1}\frac{\rho_{2}(x)}{w_{2}(x)}\right)\\
&= x^{b}(1-x)^{a}\left(((b+1)(1-x)-x(a+1))
\frac{\rho_{2}(x)}{w_{2}(x)}+x(1-x)
\left(\frac{\rho_{2}(x)}{w_{2}(x)}\right)'\right)\\
&=-D_n^{\mathrm J}x^{b}(1-x)^{a}\mathcal{P}^{a,b}_{n}(x)\mathcal{P}^{a,b}_{n-1}(x).
\end{split}
\end{equation*}
\end{proof}
\begin{remark}
For our purposes it is not helpful to compute the moment generating
function~\eqref{momentgenerating}. Although in principle one can
employ a type of fractional calculus to extract more general types of
moments from (\ref{momentgenerating}), we will see in the following
that moments for general $k$ are directly accessible with our method. 
\end{remark}
\subsection{\label{sse:JUE}Jacobi Unitary Ensemble}

Lemma~\ref{le:jacobi_den} allows us to compute the
\textit{difference} of the moments.  Then, the moments themselves can
be computed by adding all the differences.  Finally,
equation~\eqref{utransmissionmom} of theorem~\ref{the:unitary_jl} is
obtained by setting $a = \delta/2$ and $b=m-n$.

Let us define
\begin{equation}
 \label{eq:difference_jacobi}
 \Delta M^{(2)}_{\J}(k,n) = M^{(2)}_{\J}(k,n) -
 M^{(2)}_{\J}(k+1,n).
\end{equation}
\begin{proposition}
 We have
\begin{equation}
\label{unitaryjacobiresultdiff}
\Delta M^{(2)}_{\J}(k,n) = \frac{1}{k}
\sum_{j=0}^{n}\binom{k}{j}\binom{k}{j-1}U^{n, a,b}_{k,j},
\end{equation}
where
\begin{equation}
\label{gcoefficient}
U^{n,a,b}_{k,j} =
\frac{(a+b+2n-2j+k+1)(a+b+n)_{(k-j+1)}(a+n-j+1)_{(j)}
(b+n)_{(k-j+1)}}{(a+b+2n-j)_{(k+2)}(a+b+2n-j+1)_{(k)}(n+1)_{(-j)}}.
\end{equation}
If $k$ is a positive integer, equation~\eqref{unitaryjacobiresultdiff}
reduces to 
\begin{equation}
\label{unitaryjacobiresultdiff_2}
\Delta M^{(2)}_{\J}(k,n) = \frac{1}{k}
\sum_{j=0}^{\min(n,k)}\binom{k}{j}\binom{k}{j-1}U^{n, a,b}_{k,j}.
\end{equation}
\end{proposition}
\begin{proof}
Integrating by parts using equation (\ref{jacdifferentialidentity}) leads to
\begin{equation}
\label{jacobintegral}
\int_{0}^{1}x^{k}(1-x)\rho_{2}(x)dx = \frac{D_n^{\mathrm{J}}}{k}
\int_{0}^{1}x^{k+b}(1-x)^{a}\mathcal{P}^{a,b}_{n}(x)\mathcal{P}^{a,b}_{n-1}(x)dx
\end{equation}
Consider the identity
\begin{equation}
\label{jacobiconnection}
\mathcal{P}^{a,b}_{n}(x) = \sum_{j=0}^{n}\mathcal{C}_{j}^{k,n}\mathcal{P}^{a,b+k}_{j}(x),
\end{equation}
where
\begin{equation}
\mathcal{C}_{j}^{k,n} = 
\binom{k}{j}\frac{(a+n+1-j)_{(j)}(a+b+2n+1)_{(-j)}}{(a+b+2n-2j+2+k)_{(j)}(n+1)_{(-j)}}
\end{equation}
are the \textit{connection coefficients.} Inserting this
formula into (\ref{jacobintegral}) and evaluating the integrals using
orthogonality leads to
\begin{equation}
\Delta M_{\mathrm{J}_{a,b}}^{(2)}(k,n) = \frac{D^{\mathrm{J}}_{n}}{k}\sum_{j=0}^{n}\mathcal{C}_{j}^{k,n}\mathcal{C}_{j-1}^{k,n-1}h_{n-j}^{a,b+k}.
\end{equation}
Substituting the appropriate coefficients (see appendix A) gives
immediately~\eqref{unitaryjacobiresultdiff}.

When $k$ is a positive integer, the terms with $j>k$ vanish because
$\binom{k}{j}=0$ if $j>k$, leading immedietely to~\eqref{unitaryjacobiresultdiff_2}. 
\end{proof}
\begin{corollary}
The integer moments of the level density for the {\rm JUE} are 
\begin{equation}
\label{jacobiunitaryresult}
M^{(2)}_{\J}(k,n) = M^{(2)}_{\J}(1,n) - 
\sum_{j=1}^{k-1}\frac{1}{j}\sum_{i=1}^{\min(j,n)}\binom{j}{i}\binom{j}{i-1}U^{n,a,b}_{j,i}.
\end{equation}
where the first moment is
\begin{equation}
\label{eq:selberg_jacobi}
M^{(2)}_{\J}(1,n) = \frac{n(b+n)}{a+b+2n}.
\end{equation}
\end{corollary}
The first moment $M^{(2)}_{\J}(1,n)$ is an Aomoto integral.  For its
evaluation see, \textit{e.g.},~\cite{Meh04}, \S17.3.

\subsection{Laguerre Unitary Ensemble}
Since the moments of the Wigner-Smith matrix
(\ref{eq:wi_del_time_mat}) require the computation of the
integral~\eqref{eq:traces_gen} for $k<0$, we shall present formulae for
the moments of the LUE for general complex $k$.

\begin{proposition}
  Suppose that neither $b+n$ nor $b+k$ are negative integers. Then one
  has
\begin{equation}
\label{laguerremoments}
M^{(2)}_{\Lag}(k,n) =
\frac{1}{k}\sum_{j=0}^{n}\binom{k}{j}\binom{k}{j-1}
\frac{(b+n)_{(k-j+1)}}{(n+1)_{(-j)}}.
\end{equation}
\end{proposition}
\begin{proof}
 Integrating by parts the second equation in
 (\ref{differentialidentities}) gives
\begin{equation}
\label{laguerreunitaryint}
\int_{0}^{\infty}x^{k}\rho_{2}(x)dx = \frac{D_n^{\mathrm L}}{k}
\int_{0}^{\infty}x^{b+k}e^{-x}\mathcal{L}^{b}_{n}(x)\mathcal{L}^{b}_{n-1}(x)dx.
\end{equation}
For the Laguerre polynomials the connection formula is~\cite{Sze39}
\begin{equation}
\label{laguerreconnection}
\mathcal{L}^{b}_{n}(x) = 
\sum_{j=0}^{n}\mathcal{C}_{j}^{k,n}\mathcal{L}^{b+k}_{j}(x), 
\quad \text{where} \quad\mathcal{C}_{j}^{k,n} = \binom{k}{j}(n+1)_{(-j)}.
\end{equation}
Inserting formula (\ref{laguerreconnection}) into
(\ref{laguerreunitaryint}) gives
\begin{equation}
M^{(2)}_{\Lag}(k,n) = \frac{D^{\rm{L}}_{n}}{k}
\sum_{j=0}^{n}\mathcal{C}_{j}^{k,n}\mathcal{C}_{j-1}^{k,n-1}h_{n-j}^{b+k},
\end{equation}
where we evaluated the integrals using orthogonality. Using the
appropriate connection coefficients and normalisation constants
completes the proof.
\end{proof}

If $k$ is a positive integer, the binomial coefficient
$\binom{k}{j}=0$ if $j>k$, leaving only a sum with $k$ terms. Negative
moments are obtained simply by using the identity
\begin{equation*}
 \binom{-k}{j} = (-1)^{j}\binom{k+j-1}{k-1}.
\end{equation*}

\begin{corollary}
Let $k$ be a positive integer, then
\begin{equation}
\label{laguerremomentsintegerk}
M^{(2)}_{\Lag}(k,n) =
\frac{1}{k}\sum_{j=0}^{\min(n,k)}\binom{k}{j}\binom{k}{j-1}
\frac{(b+n)_{(k-j+1)}}{(n+1)_{(-j)}}.
\end{equation}
Furthermore, if $k < n + 1$ we have
\begin{equation}
 \label{laguerrenegmoms}
 M_{\mathrm{L}_b}^{(2)}(-k,n) = \frac{1}{k}
 \sum_{j=0}^{n-1}\binom{k+j}{k-1}
 \binom{k+j-1}{k-1}\frac{(b+n)_{(-k-j)}}{(n+1)_{(-j-1)}}.
\end{equation}
\end{corollary}

Equation~\eqref{eq:finite_n_mom_del} is a particular case of
formula~\eqref{laguerrenegmoms}, where $b=n$ and the scaling
introduced by the Heisenberg time $\tau_{\mathrm H}=n$ has been
taken into account.

\begin{remark}
  The appearance of the Narayana coefficients in
  \eqref{laguerremomentsintegerk} anticipates the fact that its
  leading order term as $n \to \infty$ is the $k$-\textit{th} moment
  of the Mar\v{c}enko-Pastur law~\cite{MP67}.
\end{remark}
\subsection{Gaussian Unitary Ensemble}
In this section we give the proof of equation~\eqref{gueresult1} of
theorem~\ref{the:gbetae}.  The approach is the same as for the
JUE and LUE.  The proof for $n$ odd is very similar and
we omit the details.

Integrating by parts the first formula in
(\ref{differentialidentities}) gives
\begin{equation}
\label{hermitedmoments}
M_{\mathrm G}^{(2)}(2k,n) = \int_{-\infty}^{\infty}x^{2k}\rho_{2}(x)dx
= \frac{D_n^{\mathrm H}}{2k+1}\int_{-\infty}^{\infty}x^{2k+1}e^{-x^{2}}
\mathcal{H}_{n}(x)\mathcal{H}_{n-1}(x)dx.
\end{equation}
The integral (\ref{hermitedmoments}) can be evaluated using the
Laguerre polynomials since
\begin{equation}
\label{hermitelaguerreconnection}
\mathcal{H}_{n}(x) = \mathcal{L}^{-1/2}_{n/2}(x^{2}) \quad \text{and}
\quad \mathcal{H}_{n-1}(x) = x\mathcal{L}^{1/2}_{n/2-1}(x^{2}).
\end{equation}
A change of variables then leads to 
\begin{equation}
 M_{\mathrm G}^{(2)}(2k,n)= 
 \frac{D_n^{\mathrm H}}{2k+1}
\int_{0}^{\infty}x^{k+1/2}e^{-x}\mathcal{L}^{1/2}_{n/2-1}(x)
\mathcal{L}^{-1/2}_{n/2}(x)dx.
\end{equation}
This integral is of the same type as that one in the right-hand side
of equation~\eqref{laguerreunitaryint} and can be computed in the same
way.  

\section{\label{se:orth_symp_sym}
Symplectic and Orthogonal Symmetries}

Very few non-perturbative results are available for the moments of the
densities of the eigenvalues for $\beta=1$ and $\beta=4$.  Our goal
here is to develop a novel approach that allows us to compute these
moments for all the ensembles associated with the
weights~\eqref{weights}. 

There are two possible ways of tackling this problem: the first is
through the Selberg integral; the other one is a direct computation of
the moments~\eqref{eq:traces_gen}.  The Selberg integral was very
effective in computing the moments of the transmission eigenvalues for
$\beta=2$~\cite{Nov08}; when $\beta=1$ it does not seem to produce
explicit formulae~\cite{KSS09}. It cannot be applied to $\beta=4$.

Following an approach of Dyson~\cite{Dys70}, Mehta and
Mahoux~\cite{MM91} expressed the densities for $\beta=1$ and $\beta=4$
in terms of skew-orthogonal polynomials.  Since then several articles
have attempted to improve their
formulae~\cite{NW91,TW98,Wid98,AFNvM00,GP02}.  Tracy and
Widom~\cite{TW98} and Widom~\cite{Wid98} succeeded to write such
densities as sums of $\rho_2(x)$ plus correction terms involving
orthogonal polynomials.  Building on the work of Adler and van
Moerbeke~\cite{AvM02}, Adler \textit{et al}~\cite{AFNvM00} obtained
integral representations of the correction terms.

Equation~\eqref{eq:traces_gen} presents one major challenge: for
finite $n$ it is a complicated sum involving all the orthogonal
polynomials up to $n-1$. Further integration would lead to cumbersome
formulae whose asymptotics cannot be easily extracted. Our method
relies on using the coefficients (\ref{jacobiconnection}) and
(\ref{laguerreconnection}) to expand the orthogonal polynomials in a
convenient basis, within which they are orthogonal with respect to the
perturbed weight $x^{k}w_{2}(x)$. As when $\beta=2$, this allows us to
obtain positive moments involving sums that run to the order of the
moments and not to the dimension of the ensemble. Another interesting
feature of our results is that we are able to express the moments at
$\beta=1$ in terms of the moments at $\beta=4$ plus a fairly simple
correction term. Like for unitary ensembles, our formulae are
sum of ratios of Gamma functions which may be studied
in the limit $n \to \infty$.

Since our approach is based on the results in by Adler \emph{et
  al}~\cite{AFNvM00}, we will discuss their formalism in detail. For
$\beta=1$ and $\beta=4$ a special role is played by the
skew-orthogonal polynomials. Recall that an inner product $\langle A,
B\rangle$ is referred to as \textit{skew} if $\langle A, B\rangle$ =
$-\langle B, A\rangle$. A sequence of monic polynomials
$\{q_{j}(x)\}_{j=0}^{\infty}$ are called skew-orthogonal with respect
to $\langle A, B \rangle$ if
\begin{subequations}
\label{skewdef}
\begin{align}
 \langle q_{2m}, q_{2n+1} \rangle &= - \langle q_{2n+1}, q_{2m}
 \rangle 
  = r_{m}\delta_{m,n} \\
 \langle q_{2m}, q_{2n} \rangle &= \langle q_{2m+1}, q_{2n+1}\rangle =0.
\end{align}
\end{subequations}


Let us introduce the potential $V(x)$ by defining
\begin{equation}
 w_{2}(x) = e^{-2V(x)},
\end{equation}
where $w_{2}(x)$ is the weight function of the associated unitary
ensemble. We also assume that
\begin{equation}
2V'(x) = \frac{g(x)}{f(x)}
\end{equation}
is a rational function of $x$ and take $f(x)$ to be a monic
polynomial.  Now define the \textit{modified potentials:}
\begin{subequations}
\begin{align}
 \label{beta1perturb}
 V_{1}(x) &= V(x)+\frac{1}{2}\log{f(x)}, \quad \beta=1, \\
 \label{beta4perturb}
 V_{4}(x) &= V(x)-\frac{1}{2}\log{f(x)}, \quad \beta=4.
\end{align}
\end{subequations}

Let us introduce the inner products
\begin{subequations}
\begin{align}
\label{sinnerproduct}
\langle A, B \rangle_{4} &=
\frac{1}{2}\int_{I}e^{-2V_{4}(x)}\left(A(x)B'(x)-B(x)A'(x)\right)dx
\\ 
\intertext{and}
\label{oinnerproduct}
\langle A, B \rangle_{1}& = \frac{1}{2}\int_{I}
\int_{I}e^{-V_{1}(x)-V_1(y)}\mathrm{sgn}(y-x)A(x)B(y)dxdy.
\end{align}
\end{subequations}
Associated to these inner products are two systems of monic
skew-orthogonal polynomials
$\{\tilde{q}^{(4)}_{j}(x)\}_{j=0}^{\infty}$ and
$\{\tilde{q}^{(1)}_{j}(x)\}_{j=0}^{\infty}$.  We shall denote their
skew-norms as defined in~\eqref{skewdef} by $\tilde{r}_{j}^{(4)}$ and
$\tilde{r}_{j}^{(1)}$ respectively. The tilde notation indicates that
the weight has been perturbed by the transformations
(\ref{beta1perturb}) and (\ref{beta4perturb}). 

Because the skew-orthogonality relations (\ref{skewdef}) are invariant
under the transformation $\tilde{q}_{2m+1}\to
\tilde{q}_{2m+1}+\alpha_{2m}\tilde{q}_{2m}$ for any $\alpha_{2n} \in
\mathbb{C}$, a system of skew-orthogonal polynomials is not uniquely
defined. However, they can be expressed in terms of the monic
polynomials orthogonal with respect to $w_2(x)$:
\begin{subequations}
\begin{align}
 p_{2j+1}(x)& =\tilde{q}^{(4)}_{2j+1}(x), &
 p_{2j}(x) & = \tilde{q}^{(4)}_{2j}(x)
 -\frac{c_{2j-1}}{c_{2j-2}}\tilde{q}^{(4)}_{2j-2}(x), \label{beta4skews}\\
 \tilde{q}^{(1)}_{2j}(x) & = p_{2j}(x), &
 \tilde{q}^{(1)}_{2j+1}(x) &=
 p_{2j+1}(x)-\frac{\gamma_{2j-1}}{\gamma_{2j}}p_{2j-1}(x).
\label{beta1skews}
\end{align}
\end{subequations}
For the classical orthogonal polynomials the constants in these
equations are given by
\begin{equation}
\label{ees2}
c_{n} = h_{n+1}h_{n}\gamma_{n},
\end{equation}
where
\begin{equation*}
 h_{n}\gamma_{n} =
\begin{cases}1, & \text{Hermite,}\\
\frac{1}{2},& \text{Laguerre,} \\
\frac{1}{2}(2n+a+b+2), & \text{Jacobi}.\end{cases}
\end{equation*}
We point out that the numbers $D^{\mathcal{E}}_{n}$ appearing in the
differential identities (\ref{differentialidentities}) may be
expressed in terms of $\gamma_{n}$ via $D^{\mathcal{E}}_{n} =
2\gamma_{n-1}$ for all the three ensembles. For each ensemble, we
denote the mean eigenvalue density by $\tilde{\rho}_{\beta}(x)$, where
the tilde indicates the ensemble average (\ref{eq:traces_gen}) defined
by the weight $e^{-V_{1}(x)}$ for $\beta=1$ or $e^{-2V_{4}(x)}$ for
$\beta=4$.
\begin{remark}
\label{discrepancyremark}
Before we proceed it is worth noting that the weights $e^{-V_{1}(x)}$
turn out to be exactly equal to the weights $w_{1}(x)$ in equation
\eqref{weights}. The eigenvalue densities $\tilde
\rho_{4}(x)$, however, correspond to the weights $e^{-2V_{4}(x)}$,
which are
\begin{align}
\label{weightsdiscrepancy}
e^{-2V_{4}(x)} = 
\begin{cases}
e^{-x^{2}}, & \text{{\rm Hermite},}\\
x^{b+1}e^{-x},& \text{{\rm Laguerre},}\\
x^{b+1}(1-x)^{a+1} & \text{{\rm Jacobi}.}
\end{cases}
\end{align}
These are not quite the same as the weights in \eqref{weights} for
$\beta=4$. We must make the substitution $(a,b) \to (2a,2b)$ for the
Jacobi ensemble and $b \to 2b$ for the Laguerre ensemble. In addition
there is a missing factor of $2$ in the exponentials which we take
into account by multiplying our final results for the moments by the
appropriate power of $2$. This discrepancy arises because the
symplectic ensembles are sets of self-dual $n \times n$ quaternion
matrices.  Their representation in terms of complex matrices leads to
Kramer's degeneracy, which is responsible for the normalisations in
\eqref{weightsdiscrepancy}.  Without loss of generality we shall still
use the notation $w_4(x)$.
\end{remark}

Let us introduce the $\epsilon$-transform of a suitable function
$f(t)$ by
\begin{equation}
\epsilon[f(t)](x) = \frac{1}{2}\int_{I}\mathrm{sgn}(x-t)f(t)dt.
\end{equation}
In the following, $I=(m_{1},m_{2})$ will denote the interval of
orthogonality. We have the following formulae~\cite{AFNvM00}:
\begin{subequations}
\begin{align}
\label{symplecticdens}
\tilde{\rho}_{4}(x) & = \frac{1}{2}\rho_{2}(x)_{n \to
 2n}-\frac{1}{2}\gamma_{2n-1}e^{-V_{1}(x)}p_{2n}(x)
 \int_{x}^{m_{2}}e^{-V_{1}(t)}p_{2n-1}(t)dt\\
\intertext{and (for $n$ even)}
\label{orthogonaldens}
\tilde{\rho}_{1}(x) & = \rho_{2}(x)_{n \to n-1} +
\gamma_{n-2}e^{-V_{1}(x)}p_{n-1}(x)\epsilon[p_{n-2}
(t)e^{-V_{1}(t)}](x).
\end{align}
\end{subequations}

For convenience we shall alter the representations
(\ref{symplecticdens}) and (\ref{orthogonaldens}) into a form which is
more suitable for the evaluation of the
integrals~\eqref{eq:traces_gen}. The following proposition allows us
to expand the integrals in (\ref{symplecticdens}) and
(\ref{orthogonaldens}) in terms of monic orthogonal polynomials.
\begin{proposition}
 Let $\{p_j(x)\}_{j=0}^\infty$ be the system of monic
 polynomials orthogonal with respect to $w_2(x)$. We have the
 following identities:
\begin{align}
\label{orthogonalid}
\epsilon \left[e^{-V_{1}(t)}p_{2n+2}(t)\right](x)& =-e^{-V_{4}(x)}
\sum_{j=0}^{n}e_{j,n}^{(1)}p_{2j+1}(x)+\eta_{n}^{(1)}\epsilon 
\left[e^{-V_{1}(t)}\right](x) \\
\intertext{and} 
\label{symplecticid}
\int_{x}^{m_{2}}e^{-V_{1}(t)}p_{2n+1}(t)dt &= 
-e^{-V_{4}(x)}\sum_{j=0}^{n}e_{j,n}^{(4)}p_{2j}(x),
\end{align}
where
\begin{equation}
\label{ees}
e_{j,n}^{(4)} =
\frac{h_{2n+1}}{c_{2n}}\prod_{i=j}^{n-1}\frac{c_{2i+1}}{c_{2i}},
\quad  e_{j,n}^{(1)} = \frac{h_{2n+2}}{c_{2j+1}}
\prod_{i=j+1}^{n}\frac{c_{2i}}{c_{2i+1}}, 
\quad \eta^{(1)}_{n} = \prod_{j=0}^{n}\frac{c_{2j}h_{2j+2}}{c_{2j+1}h_{2j}}.
\end{equation}
\end{proposition}
\begin{proof}
We begin from the  differential identities~\cite{AFNvM00} 
\begin{subequations}
\begin{align}
\label{beta4differentialidentity}
\frac{d}{dx}\left(e^{-V_{4}(x)}\tilde{q}^{(4)}_{2n}(x)\right) & =
\frac{c_{2n}}{h_{2n+1}}e^{-V_{1}(x)}p_{2n+1}(x), \\
\label{beta1differentialidentity}
\frac{d}{dx}\left(e^{-V_{4}(x)}\tilde{q}^{(4)}_{2n+1}(x)\right) & =
e^{-V_{1}(x)}\left(\frac{c_{2n}}{h_{2n}}
 p_{2n}(x)-\frac{c_{2n+1}}{h_{2n+2}}p_{2n+2}(x)\right).
\end{align}
\end{subequations}
We also need $e^{-V_{4}(m_{1})} = e^{-V_{4}(m_{2})}=0$, which can
be easily checked from (\ref{weightsdiscrepancy}). 

We first derive~\eqref{symplecticid}. Integrating equation
(\ref{beta4differentialidentity}) between $x$ and $m_{2}$ gives
\begin{equation}
\label{symplecticid2}
\int_{x}^{m_{2}}e^{-V_{1}(t)}p_{2n+1}(t)dt = 
-\frac{h_{2n+1}}{c_{2n}}e^{-V_{4}(x)}q^{(4)}_{2n}(x) 
= -e^{-V_{4}(x)}\sum_{j=0}^{n}e_{j,n}^{(4)}p_{2j}(x),
\end{equation}
where the last equality was obtained by iteratively solving equation
(\ref{beta4skews}) for $\tilde{q}^{(4)}_{2n}(x)$. 

In order to derive (\ref{orthogonalid}), we start by integrating
equation (\ref{beta1differentialidentity}) between $x$ and $m_{2}$:
\begin{equation}
\label{between0infty}
\int_{x}^{m_{2}}e^{-V_{1}(t)}p_{2n+2}(t)dt=
\frac{c_{2n}h_{2n+2}}{c_{2n+1}h_{2n}}\int_{x}^{m_{2}}
e^{-V_{1}(t)}p_{2n}(t)dt
-\frac{h_{2n+2}}{c_{2n+1}}e^{-V_{4}(x)}q^{(4)}_{2n+1}(x).
\end{equation}
Integrating (\ref{beta1differentialidentity}) between $m_{1}$ and $x$
and subtracting the result from (\ref{between0infty}) gives an
equation for the $\epsilon$-transform
\begin{equation}
 \epsilon\left[e^{-V_{1}(t)}p_{2n+2}(t)\right] 
= \frac{c_{2n}h_{2n+2}}{c_{2n+1}h_{2n}}
\epsilon\left[e^{-V_{1}(t)}p_{2n}(t)\right]
-\frac{h_{2n+2}}{c_{2n+1}}e^{-V_{4}(x)}p_{2n+1}(x),
\end{equation}
where we used that $\tilde q^{(4)}_{2j+1} =
p_{2j+1}(x)$. Iterating this equation $n$ times leads to
(\ref{orthogonalid}).
\end{proof}
\begin{remark}
 The coefficients $e^{(4)}_{j,n}$, $e^{(1)}_{j,n}$ and
 $\eta^{(1)}_{n}$ are tabulated in appendix~A for each ensemble.
\end{remark}
\begin{corollary}
\label{eigenvaluedensitymassaged}
We have the following representations for the eigenvalue densities:
\begin{subequations}
\label{eq:densities_1_4}
\begin{align}
\label{beta4densityorthogs}
\tilde{\rho}_{4}(x) &= \frac{1}{2}\rho_{2}(x)_{n \to 2n}
-\frac{1}{2}\gamma_{2n-1}e^{-2V(x)}\sum_{j=0}^{n-1}e_{j,n-1}^{(4)}
p_{2j}(x)p_{2n}(x), \\
\tilde{\rho}_{1}(x) &= \rho_{2}(x)_{n \to n-1}-\gamma_{n-2}e^{-2V(x)}\sum_{j=0}^{n/2-2}e^{(1)}_{j,n/2-2}p_{2j+1}(x)p_{n-1}(x) \label{beta1densityorthogs}\\
&\quad +\gamma_{n-2}e^{-V_{1}(x)}p_{n-1}(x)
\eta^{(1)}_{n/2-2}\epsilon\left[e^{-V_{1}(t)}\right](x).
\notag
\end{align}
\end{subequations}
\end{corollary}
\begin{proof}
 Substituting the integration identities (\ref{symplecticid}) and
 (\ref{orthogonalid}) into (\ref{symplecticdens}) and
 (\ref{orthogonaldens}) respectively, and using
 $V_{1}(t)+V_{4}(t)=2V(t)$, gives (\ref{beta4densityorthogs}) and
 (\ref{beta1densityorthogs}).
\end{proof}

We will see in \S\ref{se:symp_ens} and \S\ref{sec:orth_ens} that
formulae~\eqref{eq:densities_1_4} are particularly suited to our
purposes. A key feature of these representations is that they are
expressed solely in terms of the weight function $e^{-2V(x)}$ and the
corresponding monic orthogonal polynomials. For the orthogonal
ensembles, there is an additional term involving the
$\epsilon$-transform of the weight $e^{-V_{1}(t)}$, which is related
to the error function, incomplete gamma function or incomplete beta
function depending on the ensemble in question. We shall compute the
moments for $\beta=1$ and $\beta=4$ by combining the
representations~\eqref{eq:densities_1_4} with a variant of the
technique used for unitary ensembles.

\section{\label{se:symp_ens}Symplectic Ensembles}

The purpose of this section is to compute the integrals
\begin{equation}
\label{beta4integrals}
\tilde{M}^{(4)}_{\mathcal{E}}(k,n) = \int_{I}x^{k}\tilde{\rho}_{4}(x)dx
\end{equation}
for each ensemble $\mathcal{E}$ defined by the weights
(\ref{weights}). 

Inserting the representation (\ref{beta4densityorthogs}) into equation
(\ref{beta4integrals}) leads to two integrals: the first one contains
the mean eigenvalue density of a unitary ensemble, which was computed
in \S\ref{sec:unitary_ens}; the second one involves orthogonal
polynomials. More explicitly it is given by
\begin{equation}
\label{tildesymplecticintegrals}
\tilde{S}_{\mathcal{E}}(k,n) = \gamma_{2n-1}
\frac{1}{2}\sum_{j=0}^{n-1}e^{(4)}_{j,n-1}\int_{I}x^{k}e^{-2V(x)}
p_{2j}(x)p_{2n}(x)dx.
\end{equation}
We know how to evaluate these integrals as they are exactly of the
type that appeared in equations (\ref{jacobintegral}),
(\ref{laguerreunitaryint}) and \eqref{hermitedmoments}.  We write the
polynomials $p_{n}(x)$ in a basis which is orthogonal with
respect to the perturbed weight $x^{k}e^{-2V(x)}$; then, we can use
the orthogonality of the polynomials to write the integral in
(\ref{tildesymplecticintegrals}) as a single sum involving the
connection coefficients (\ref{jacobiconnection}) and
(\ref{laguerreconnection}). Eventually, the moments for $\beta=4$
become
\begin{equation}
 \label{tildesymplecticmoments}
\tilde{M}^{(4)}_{\mathcal{E}}(k,n) = M^{(2)}_{\mathcal{E}}(k,2n)
-\tilde{S}_{\mathcal{E}}(k,n).
\end{equation}

As in \S\ref{se:orth_symp_sym} the tilde notation indicates
quantities that differ from the integrals~\eqref{integrals} by a
factor discussed in remark~\ref{discrepancyremark}.

\subsection{Jacobi Symplectic Ensemble}
We now compute $S_{\J}(k,n)$ for complex $k$.
Equation~\eqref{eq:mom_b4_del} is then obtained by restricting $k$ to
be a positive integer  and setting $a=\delta/4 - 1/2$ and $b=m-n$. 

\begin{proposition}
\label{symplecticjacobiprop}
We have
\begin{equation}
\label{jacsprop}
S_{\J}(k,n) = \sum_{j=1}^{n}\sum_{i=0}^{2n-2j}
\binom{k}{i+2j}\binom{k}{i}S^{a,b}_{i,j}(k,n),
\end{equation}
where the coefficient $S^{a,b}_{i,j}(k,n)$ is given by
\begin{equation}
\label{generaljacobicoefficient}
\begin{split}
 S^{a,b}_{i,j}(k,n) & = \frac{2^{4j-3}(2a+2n-i-2j+1)_{(i)}
   (2b+2n)_{(k-i-2j+1)}(2a+2b+2n)_{(k-i-2j+1)}}%
 {(2n-2j+1)_{(-i)}(n+1)_{(-j)}(a+n+1)_{(-j)}(b+n)_{(1-j)}(a+b+n)_{(1-j)}} \\
 &\quad \times
 \frac{(2a+2b+4n-4j+1)(2a+2b+4n-2i-4j+k+1)}%
{(2a+2b+4n-i-2j+1)_{(1+k)}(2a+2b+4n-i-4j+1)_{(1+k)}}.
\end{split}
\end{equation}
\end{proposition}
\begin{proof}
By (\ref{tildesymplecticintegrals}) we have
\begin{equation}
\label{jacsdef}
\tilde{S}_{\J}(k,n)  = \frac{\gamma_{2n-1}}{2}\sum_{j=0}^{n-1}e_{j,n-1}^{(4)}\int_{0}^{1}x^{b+k}(1-x)^{a}\mathcal{P}^{a,b}_{2n}(x)\mathcal{P}^{a,b}_{2j}(x)dx.
\end{equation}
Inserting the connection formula (\ref{jacobiconnection}) into the
integrand leads to
\begin{equation}
\begin{split}
 \tilde{S}_{\J}(k,n)&  = \frac{\gamma_{2n-1}}{2}
\sum_{j=0}^{n-1}e_{j,n-1}^{(4)}\sum_{i=0}^{2j}\mathcal{C}_{i}^{k,2j}\sum_{p=0}^{2n}\mathcal{C}_{p}^{k,2n}\int_{0}^{1}x^{b+k}(1-x)^{a}\mathcal{P}^{a,b+k}_{2n-p}(x)\mathcal{P}^{a,b+k}_{2j-i}(x)dx\\
 &= \frac{\gamma_{2n-1}}{2}\sum_{j=0}^{n-1}e_{j,n-1}^{(4)}\sum_{i=0}^{2j}\mathcal{C}_{i}^{k,2j}\sum_{p=0}^{2n}\mathcal{C}_{p}^{k,2n}h_{2j-i}^{a,b+k}\delta_{2n-p,2j-i}\\
 &=\frac{\gamma_{2n-1}}{2}\sum_{j=1}^{n}\sum_{i=0}^{2n-2j}e_{n-j,n-1}^{(4)}h_{2n-2j-i}^{a,b+k}\mathcal{C}_{i}^{k,2n-2j}\mathcal{C}_{i+2j}^{k,2n}. \label{jacsresult}
\end{split}
\end{equation}
To obtain the above expression we have applied the orthogonality of
the Jacobi polynomials and then rearranged the indices in the
sum. Substituting the coefficients $\mathcal{C}_j^{k,n}$, $h^{a,b+k}_n$ and
$e^{(4)}_{j,k}$ (see appendix A) and replacing $(a,b)\to (2a,2b)$
completes the proof.
\end{proof}
\begin{remark}
 The complexity of the expression \eqref{generaljacobicoefficient} is
 mainly due to formula \eqref{jacobiconnection}, which relates Jacobi
 polynomials of different weights. The Laguerre ensemble is slightly
 simpler, because the associated coefficients
 \eqref{laguerreconnection} and normalizations $h_{j}$ are more
 concise.
\end{remark}
\subsection{Laguerre Symplectic Ensemble}
\begin{proposition}
\label{pro:symp_4}
Suppose that neither $2b+k$ nor $b+n$ are negative
integers. Then, we have
\begin{equation}
\label{lagsprop}
S_{\Lag}(k,n) = \sum_{j=1}^{\lfloor n \rfloor}\sum_{i=0}^{2n-2j}\binom{k}{i}\binom{k}{i+2j}\frac{(2b+2n)_{(k-i-2j+1)}(2n-i-2j+1)_{(i)}}{2^{k-2j+2}(n+1)_{(-j)}(b+n)_{(1-j)}}.
\end{equation}
\end{proposition}
\begin{proof}
 From (\ref{tildesymplecticintegrals}) we have
\begin{equation}
\label{lagsdef}
\tilde{S}_{\Lag}(k,n) = \frac{\gamma_{2n-1}}{2}
\sum_{j=0}^{n-1}e_{j,n-1}^{(4)}\int_{0}^{\infty}x^{b+k}e^{-x}
\mathcal{L}^{b}_{2n}(x)\mathcal{L}^{b}_{2j}(x)dx.
\end{equation}
Proceeding as in the proof of proposition~\ref{symplecticjacobiprop}
we obtain
\begin{equation}
\label{lagsresult}
\begin{split}
\tilde{S}_{\Lag}(k,n) & = \frac{\gamma_{2n-1}}{2}\sum_{j=0}^{n-1}e_{j,n-1}^{(4)}\sum_{i=0}^{2j}\mathcal{C}_{i}^{k,2j}\sum_{p=0}^{2n}\mathcal{C}_{p}^{k,2n}\int_{0}^{\infty}x^{b+k}e^{-x}\mathcal{L}^{b+k}_{2n-p}(x)\mathcal{L}^{b+k}_{2j-i}(x)dx\\
&= \frac{\gamma_{2n-1}}{2}\sum_{j=0}^{n-1}e_{j,n-1}^{(4)}\sum_{i=0}^{2j}\mathcal{C}_{i}^{k,2j}\sum_{p=0}^{2n}\mathcal{C}_{p}^{k,2n}h_{2j-i}^{b+k}\delta_{2n-p,2j-i}\\
&=\frac{\gamma_{2n-1}}{2}\sum_{j=1}^{n}\sum_{i=0}^{2n-2j}e_{n-j,n-1}^{(4)}\mathcal{C}_{i}^{k,2n-2j}\mathcal{C}_{i+2j}^{k,2n}h_{2n-2j-i}^{b+k} 
\end{split}
\end{equation}
By replacing $b \to 2b$ and multiplying both sides of this equation by
$2^{-k}$ gives the statement of the proposition.
\end{proof}

Both the propositions~{\rm \ref{symplecticjacobiprop}} and
{\rm \ref{pro:symp_4}} hold for complex values of $k$, except where
the integrals~\eqref{integrals} diverge. In particular, the moments of
the proper delay times \eqref{eq:mom_time_del} are expressed in terms
of negative moments of the Laguerre ensemble and can be obtained from
\eqref{lagsprop} using the identity
\begin{equation}
\binom{-k}{j}\binom{-k}{i+2j} = \binom{k+j-1}{k-1}\binom{k+i+2j-1}{k-1}
\end{equation}
and setting $b=n+1$ in \eqref{lagsprop}. Thus, we arrive at the
following:
\begin{corollary}
\label{co:lag_jac_int_4}
Let $k$ be a positive integer.  We have
\begin{align}
\label{eq:jacobi_symp_int}
 S_{\J}(k,n) & = \sum_{j=1}^{\min(\lfloor n \rfloor,\lfloor k/2 \rfloor)}
\sum_{i=0}^{\min(2n-2j,k-2j)}\binom{k}{i+2j}\binom{k}{i}S^{a,b}_{i,j}(k,n)\\
 S_{\Lag}(k,n) & = \sum_{j=1}^{\min(\lfloor n \rfloor,\lfloor k/2
   \rfloor)}\sum_{i=0}^{\min(2n - 2j,k-2j)}
 \binom{k}{i+2j}\binom{k}{i}\notag \\
\label{eq:laguerre_symp_int}
& \quad \times 
\frac{(2b+2n)_{(k-i-2j+1)}(2n-i-2j+1)_{(i)}}{2^{k-2j+2}(n+1)_{(-j)}(n+b)_{(1-j)}}.
\end{align}
Furthermore, if $ k< 2n +1$
\begin{equation}
\label{eq:lag_symp_neg_int}
\begin{split}
 S_{\Lag}(-k,n) & = 
\sum_{j=1}^{\lfloor n \rfloor}\sum_{i=0}^{2n-2j}
\binom{k+j-1}{k-1}\binom{k+i+2j-1}{k-1}\\
& \quad \times \frac{(2b+2n)_{(-k-i-2j+1)}(2n-i-2j+1)_{(i)}}%
{2^{-k-2j+2}(n+1)_{(-j)}(n+b)_{(1-j)}}.
\end{split}
\end{equation}
\end{corollary}
\begin{remark}
 The combinatorial aspect of the sums in this corollary arises
 directly from the connection coefficients $\mathcal{C}_{i}$ and
 $\mathcal{C}_{i+2j}$ appearing in \eqref{lagsresult} and
 \eqref{jacsresult}, leading to the binomial coefficients. Due to
 these binomial coefficients, the sums (\ref{eq:jacobi_symp_int})
 only go up to $k$ and their complexity does not increase as $n$
 grows. This therefore permits an investigation of the
 asymptotics. The $n \to \infty$ analysis of
 (\ref{eq:lag_symp_neg_int}) leads to certain infinite sums which
 also turn out to be tractable. A similar remark holds for the GSE.
\end{remark}

\subsection{Gaussian Symplectic Ensemble}
\label{se:gse}
In the Gaussian case, only the even moments are different from zero.
For simplicity we only state the results for integer $k$.
\begin{proposition}
We have
\begin{equation}
\label{hersprop}
\begin{split}
 S_{\G}(2k,n) & =\frac{\Gamma(n+1)\Gamma(n)}{2^{k}\sqrt{\pi}\Gamma(2n)4^{1-n}}
 \sum_{j=1}^{\min(n,k)}\sum_{i=0}^{\min(n-j,k-j)}
\binom{k}{i}\binom{k}{i+j}(n-i-j+1)_{(k-1/2)}
\end{split}
\end{equation}
\end{proposition}
\begin{proof}
 The integrals (\ref{tildesymplecticintegrals}) give
\begin{equation}
\label{hersdef}
\tilde{S}_{\G}(2k,n) =
\frac{\gamma_{2n-1}}{2}\sum_{j=0}^{n-1}e_{j,n-1}^{(4)}
\int_{-\infty}^{\infty}x^{2k}e^{-x^{2}}\mathcal{H}_{2n}(x)\mathcal{H}_{2j}(x)dx.
\end{equation}
Applying formula~\eqref{hermitelaguerreconnection} leads to
\begin{equation}
 \tilde{S}_{\G}(2k,n) =
 \frac{\gamma_{2n-1}}{2}\sum_{j=0}^{n-1}e_{j,n-1}^{(4)}
 \int_{-\infty}^{\infty}x^{2k}e^{-x^{2}}
\mathcal{L}_{n}^{-1/2}(x^{2})\mathcal{L}_{j}^{-1/2}(x^{2})dx.
\end{equation}
Changing variables and inserting the connection 
formula (\ref{laguerreconnection}) results in
\begin{equation*} 
\tilde{S}_{G}(2k,n) = \frac{\gamma_{2n-1}}{2}
\sum_{j=0}^{n-1}e_{j,n-1}^{(4)}\sum_{i=0}^{j}\mathcal{C}_{i}^{k,j}\sum_{p=0}^{n}\mathcal{C}_{p}^{k,n}\int_{0}^{\infty}x^{k-1/2}e^{-x}\mathcal{L}_{n-p}^{k-1/2}(x)\mathcal{L}_{j-i}^{k-1/2}(x)dx.
\end{equation*}
Using the orthogonality of the Laguerre polynomials gives the double sum
\begin{equation}
\label{eq:last_symp}
\tilde{S}_{G}(2k,n) = \frac{\gamma_{2n-1}}{2}
\sum_{j=1}^{n}\sum_{i=0}^{n-j}e_{n-j,n-1}^{(4)}
\mathcal{C}_{i}^{k,n-j}\mathcal{C}_{i+j}^{k,n}h^{k-1/2}_{n-i-j}.
\end{equation}
It is worth emphasising that the coefficients
$\mathcal{C}_{i}^{k,n-j}$, $\mathcal{C}_{i+j}^{k,n}$ and $h_{n}$ are
those of the Laguerre polynomials, while the coefficient
$e_{n-j,n-1}^{(4)}$ is related to the Hermite polynomials.  Finally,
multiplying equation~\eqref{eq:last_symp} by $2^{-k}$, as discussed in
remark~\ref{discrepancyremark}, completes the proof.
\end{proof}
\section{\label{sec:orth_ens}Orthogonal Ensembles}

In this section we compute the moments for $\beta=1$.
Theorem~\ref{th:beta1moments} and equation~\eqref{beta1gauss} of
Theorem~\ref{the:gbetae} are corollaries of the results we prove
here.  For simplicity we assume that $n$ is an even integer.

The main task is to compute the integral
\begin{equation}
\label{orthogonalintegrals}
M_{\mathcal{E}}^{(1)}(k,n) = \int_{I}x^{k}\tilde{\rho}_{1}(x)dx.
\end{equation}
When $\beta=1$ the density $\tilde \rho_1(x)$ coincides with
$\rho_1(x)$, so the integrals~\eqref{orthogonalintegrals} coincide
with the averages~\eqref{integrals}.  

Substituting the representation (\ref{beta1densityorthogs}) into
(\ref{orthogonalintegrals}), we are left to compute three integrals:
the first one gives the moments of the corresponding unitary ensemble;
the second one is closely related to the quantity
$\tilde{S}_{\mathcal{E}}(k,n)$ discussed in \S\ref{se:symp_ens},
namely
\begin{equation}
\label{ointegrals}
O_{\mathcal{E}}(k,n) = \gamma_{n-2}\sum_{j=0}^{n-1}e^{(1)}_{j,n/2-2}
\int_{I}x^{k}e^{-2V(x)}\mathcal{P}_{2j+1}(x)\mathcal{P}_{n-1}(x)dx;
\end{equation}
the last one arises from the $\epsilon$-transform in equation
(\ref{beta1densityorthogs}), \textit{i.e.}
\begin{equation}
\label{incompleteintegral}
I_{\mathcal{E}}(k,n) =\gamma_{n-2}\eta^{(1)}_{n/2-2}
\int_{I}x^{k}\mathcal{P}_{n-1}(x)e^{-V_{1}(x)}
\epsilon\left[e^{-V_{1}(t)}\right](x)dx.
\end{equation}
Therefore, the moments for $\beta=1$ may be expressed as
\begin{equation}
M^{(1)}_{\mathcal{E}}(k,n) = M^{(2)}_{\mathcal{E}}(k,n-1)
-O_{\mathcal{E}}(k,n)+I_{\mathcal{E}}(k,n)
\end{equation}
In this section we focus on the integrals $O_{\mathcal{E}}(k,n)$ and
$I_{\mathcal{E}}(k,n)$.

\subsection{A Duality between $\beta=1$ and $\beta=4$}
We first discuss a remarkable duality between the quantities
$M^{(2)}_{\mathcal{E}}(k,n-1)-O_{\mathcal{E}}(k,n)$ and the moments of
the symplectic ensembles $M^{(4)}_{\mathcal{E}}(k,n)$, where $n$ now
can assume half integer values.  Such moments are well defined (see
equation~\eqref{tildesymplecticmoments} and remark~\ref{re:half_int}).   Similar
dualities have appeared in the literature
before~\cite{Led04,For09,Des09}.

\begin{lemma}
\label{pr:duality}
 Let $n$ be an even integer.  We have the following dualities:
\begin{subequations}
\begin{align}
 M^{(1)}_{\Lag}(k,n)& = 2^{1+k}M^{(4)}_{\mathrm{L}_{b/2}}(k,(n-1)/2)+I_{\Lag}(k,n), \label{laguerreduality}\\
 M^{(1)}_{\J}(k,n) &=
 2M^{(4)}_{\mathrm{J}_{a/2,b/2}}(k,(n-1)/2)+I_{\J}(k,n). \label{jacobiduality}
\end{align}
\end{subequations}
\end{lemma}
\begin{proof}
 Firstly, by equation (\ref{tildesymplecticmoments}) we observe
 that
\begin{equation}
2^{1+k}M^{(4)}_{\mathrm{L}_{b/2}}(k,(n-1)/2) = M^{(2)}_{\Lag}(k,n-1) - 2\tilde{S}_{\mathrm{L}_{b}}(k,(n-1)/2).
\end{equation}
Thus, it is sufficient to check that
$2\tilde{S}_{\mathrm{L}_{b/2}}(k,(n-1)/2) = O_{\Lag}(k,n)$.  

A direct computation shows that 
\begin{equation}
\label{lagoresult}
\begin{split}
O_{\Lag}(k,n) & = 
\gamma_{n-2}\sum_{j=0}^{n/2-2}e_{j,n/2-2}^{(1)}\int_{0}^{\infty}x^{b+k}e^{-x}
\mathcal{L}^{b}_{n-1}(x)\mathcal{L}^{b}_{2j+1}(x)dx\\
&=\gamma_{n-2}\sum_{j=0}^{n/2-2}e_{j,n/2-2}^{(1)}\sum_{i=0}^{2j+1}\mathcal{C}_{i}^{k,2j+1}\sum_{p=0}^{n-1}\mathcal{C}_{p}^{k,n-1}h_{2j+1-i}^{b+k}\delta_{2j+1-i,n-1-p}\\
&=\gamma_{n-2}\sum_{j=1}^{n/2-1}\sum_{i=0}^{n-2j-1}e_{n/2-1-j,n/2-2}^{(1)}\mathcal{C}_{i}^{k,n-2j-1}\mathcal{C}_{i+2j}^{k,n-1}h_{n-1-2j-i}^{b+k}.
\end{split}
\end{equation}
From (\ref{lagsresult}) we see that
\begin{equation}
\label{lagoresult_2}
2\tilde{S}_{\Lag}(k,(n-1)/2)=\gamma_{n-2}\sum_{j=1}^{n/2-1}\sum_{i=0}^{n-2j-1}e_{n/2-1/2-j,n/2-3/2}^{(4)}\mathcal{C}_{i}^{k,n-2j-1}\mathcal{C}_{i+2j}^{k,n-1}h_{n-1-2j-i}^{b+k}.
\end{equation}
From equation (\ref{laguerreees}) we have that
\begin{equation}
 \label{eq:test_s}
  e_{n/2-1/2-j,n/2-3/2}^{(4)}=e_{n/2-1-j,n/2-2}^{(1)}.
\end{equation}
Thus, the right-hand sides of equations~\eqref{lagoresult} and
\eqref{lagoresult_2} coincide.

In the proof of the duality~\eqref{jacobiduality} one has to show that
\begin{equation*}
 2\tilde{S}_{\mathrm{J}_{a/2,b/2}}(k,(n-1)/2)
= O_{\J}(k,n).
\end{equation*}
The strategy is the same as for the Laguerre ensemble and we omit the
computation.
\end{proof}
We are now left with the task of computing the integrals
(\ref{incompleteintegral}). When $k$ is a positive integer, we find a
single sum containing $k$ terms for each ensemble; when $k$ is
negative the sums go up to order of the matrix dimension.  


\subsection{\label{sse:inc_int}Incomplete Integrals --- Positive Moments}

We now assume that $k$ is a positive integer and focus on the Laguerre and
Jacobi ensembles. 

\begin{lemma}
We have
\begin{equation}
\label{eq:Ib}
I_{\Lag}(k,n) = 2^{k}\sum_{j=0}^{\min(n/2-1,k)}
\binom{2k}{2j}\frac{(\frac{1}{2}(b+n))_{(k-j)}}{(\frac{1}{2}(1+n))_{(-j)}} 
+ \tilde{\phi}_{k,n}^{\mathrm{L}}
\end{equation} 
and
\begin{equation}
\label{eq:Iab}
\begin{split}
I_{\J}(k,n) = & 4^{k}\sum_{j=0}^{\min(n/2 -1,k)}\binom{2k}{2j}
\frac{(a+b+2n-4j-1+2k)(\frac{1}{2}(a+b+n))_{(k-j)}(\frac{1}{2}(b+n))_{(k-j)}}
{(a+b+2n-2j-1)_{(2k+1)}(\frac{1}{2}(a+n+1))_{(-j)}(\frac{1}{2}(1+n))_{(-j)}} \\
&+\phi_{k,n}^{\mathrm{J}},
\end{split}
\end{equation}
where 
\begin{equation}
\label{integralphi}
\tilde{\phi}_{k,n}^{\mathrm{L}} = \frac{\Gamma(n/2+b/2-1/2)}{\Gamma(n/2)\Gamma(b+n-1)}\sum_{j=1}^{k}\frac{\Gamma(j+k)\Gamma(b+j+k)2^{-j}}{\Gamma(j+k-n+1)\Gamma(b/2+1/2+j)}
\end{equation}
and
\begin{equation}
\label{jacobiphi}
\begin{split}
\phi^{\mathrm{J}}_{k,n} & = \sum_{j=1}^{k}\frac{2^{a+1}\Gamma(a/2+b/2+j)
\Gamma(a/2+b/2+1/2+n/2)}{\Gamma(a+b+n+j+k)\Gamma(j+k+1-n)\Gamma(b/2+1/2+j)} \\
& \quad \times \frac{\Gamma(a/2+1/2+n/2)\Gamma(b+k+j)
\Gamma(j+k)}{\Gamma(n/2+b/2)\Gamma(a/2+1/2)\Gamma(n/2)}. 
\end{split}
\end{equation}
Furthermore, if $n > 2k$,
$\tilde{\phi}_{k,n}^{\mathrm{L}}=\phi^{\mathrm{J}}_{k,n}=0$.
\end{lemma}
\begin{proof}
  We begin with the proof of formula~\eqref{eq:Ib}. Equation
  (\ref{incompleteintegral}) becomes
\begin{equation}
\label{ilagdef}
I_{\Lag}(k,n) =\gamma_{n-2}\eta^{(1)}_{n/2-2}
\int_{0}^{\infty}x^{k+(b-1)/2}\mathcal{L}^{b}_{n-1}(x)
\epsilon\left[t^{(b-1)/2}e^{-t/2}\right](x)dx.
\end{equation}
The $\epsilon$-transform appearing in the right-hand side is 
the difference of the two incomplete Gamma functions
\begin{equation}
 \label{eq:inc_gam}
 \gamma(a,z) = \int_0^zt^{a-1}e^{-t}dt \quad \text{and} \quad
 \Gamma(a,z) = \int_z^{\infty}t^{a-1}e^{-t}dt.
\end{equation}
The main idea here is to expand them in a sum of incomplete Gamma
functions whose weight has been perturbed by a factor $t^{k}$. To this
end we insert
\begin{equation}
\label{incompletegammaidentity}
\epsilon\left[t^{\frac{b-1}{2}}e^{-t/2}\right](x) 
= \sum_{j=1}^{k}d_{j}^{(b+1)/2}
x^{\frac{b-1}{2}+j}e^{-x/2}+\frac{1}{2}d_{k}^{(b+1)/2}\epsilon\left[t^{\frac{b-1}{2}+k}e^{-t/2}\right](x),
\end{equation}
where $d_{j}^{b} = 2^{1-j}\frac{\Gamma(b)}{\Gamma(b+j)}$.  When $k=1$
this identity  can be found in~\cite{AS72}, equations~(6.5.21)
and~(6.5.23); the formula for general $k$ is obtained by iteration.

This leads to two integrals.  The first one is
\begin{equation}
\label{zerolag}
\begin{split}
 \tilde{\phi}_{k,n}^{\mathrm L} & =
 \eta^{(1)}_{n/2-2}\gamma_{n-2}\sum_{j=1}^{k}
 \int_{0}^{\infty}d_{j}^{(b+1)/2}x^{b+j+k-1}e^{-x}\mathcal{L}_{n-1}^{b}(x)dx\\
 &=\eta^{(1)}_{n/2-2}\gamma_{n-2}\sum_{j=1}^{k}\sum_{i=0}^{j+k-1}\int_{0}^{\infty}\mathcal{C}_{i}^{j+k-1,n-1}d_{j}^{(b+1)/2}x^{b+j+k-1}e^{-x}
 \mathcal{L}_{n-1-i}^{b+j+k-1}(x)dx,
\end{split}
\end{equation}
where we inserted the connection formula
(\ref{laguerreconnection}). Because of the orthogonality of the
Laguerre polynomials the only contribution to the inner sum occurs at
$i=n-1$; furthermore, since $\mathrm{max}(i)=2k-1$, we have that
$\tilde{\phi}_{k,n}^{\mathrm L}=0$ if $n >2k$. If $n \le 2k$,
inserting the appropriate coefficients (see appendix A) and using the
duplication formula (see~\cite{AS72}, \S6)
\begin{equation}
 \label{eq:duplication_for}
 \Gamma(2z) = \pi^{-1/2}2^{2z -1}\Gamma\left(z\right)
\Gamma\left(z + \tfrac{1}{2}\right)
\end{equation}
gives~\eqref{integralphi}.

The remaining non-trivial integral is 
\begin{subequations}
\begin{align}
\psi_{k,n} & = C\int_{0}^{\infty}x^{\frac{b-1}{2}+k}e^{-x/2}\mathcal{L}_{n-1}^{b}(x)\epsilon\left[t^{\frac{b-1}{2}+k}e^{-t/2}\right](x)dx \notag\\
&=C\sum_{j=0}^{\min(n-1,2k)}\int_{0}^{\infty}x^{\frac{b-1}{2}+k}e^{-x/2}\mathcal{C}_{j}^{2k,n-1}\mathcal{L}_{n-1-j}^{b+2k}(x)\epsilon\left[t^{\frac{b-1}{2}+k}e^{-t/2}\right](x)dx \label{secondline}\\
&=C\sum_{j=0}^{\min(n/2-1,k)}\mathcal{C}_{2j}^{2k,n-1}\int_{0}^{\infty}x^{\frac{b-1}{2}+k}e^{-x/2}\mathcal{L}_{n-2j-1}^{b+2k}(x)\epsilon\left[t^{\frac{b-1}{2}+k}e^{-t/2}\right](x)dx, \label{integratemebyparts}
\end{align}
\end{subequations}
where $C=\eta^{(1)}_{n/2-2}\gamma_{n-2}d_{k}^{(b+1)/2}/2$. To obtain
(\ref{secondline}) we used the connection formula, while
(\ref{integratemebyparts}) follows from the fact that the
contributions to the sum with odd indices vanish due to the
skew-orthogonality constraints (\ref{skewdef}). 

Now we integrate (\ref{integratemebyparts}) by parts using the
identity (\ref{symplecticid}) and the formula
\begin{equation}
\frac{d}{dx}\epsilon\left[t^{\frac{b-1}{2}+k}e^{-t/2}\right](x) 
= -x^{\frac{b-1}{2}+k}e^{-x/2}.
\end{equation}
This leads to
\begin{equation}
\begin{split}
\psi_{k,n}&=C\sum_{j=0}^{\min(n/2-1,k)}\sum_{i=0}^{n/2-j-1}\mathcal{C}_{2j}^{2k,n-1}e_{i,n/2-j-1}^{(4),b+2k}\int_{0}^{\infty}x^{b+2k}e^{-x}\mathcal{L}^{b+2k}_{2i}(x)dx\\
&=C\sum_{j=0}^{\min(n/2-1,k)}\mathcal{C}_{2j}^{2k,n-1}e_{0,n/2-j-1}^{(4),b+2k}h^{b+2k}_{0}. \label{laguerrefinalline}
\end{split}
\end{equation}
Inserting all the relevant formulae for the orthogonality norms
and connection coefficients gives~\eqref{eq:Ib}.

We sketch the proof of~\eqref{eq:Iab} as it follows a similar pattern;
we only emphasise the differences.  For the Jacobi ensemble the
$\epsilon$-transform is expressed in terms of incomplete Beta
functions; thus, we replace (\ref{incompletegammaidentity}) with the
following identity:
\begin{equation}
\label{lastmin1}
\epsilon[t^{b}(1-t)^{a}](x) = 
\sum_{j=1}^{k}d_{j}^{a,b}x^{b+j}(1-x)^{a+1}+(a+b+k+1)
d_{k}^{a,b}\epsilon[t^{b+k}(1-t)^{a}](x),
\end{equation}
where 
\begin{equation}
d_{j}^{a,b} = \frac{\Gamma(a+b+j+1)\Gamma(b+1)}{\Gamma(a+b+2)\Gamma(b+1+j)}.
\end{equation}
Equation~\eqref{lastmin1} can be obtained by iteration from
formulae~(26.5.15) and~(25.5.16) in~\cite{AS72}. Proceeding as for the
Laguerre ensemble and using formula~\eqref{eq:duplication_for}
gives~\eqref{eq:Iab}.
\end{proof}

\subsection{Incomplete Integrals --- Negative Moments}
When the moments are negative, we focus only on the physically
interesting case of the Laguerre ensemble, which leads to moments of
the proper delay times.

When the moments are positive the correction terms
$\tilde{\phi}^{\mathrm{L}}_{k,n}$ and $\phi^{\mathrm{J}}_{k,n}$ vanish
if $n>2k$. Now we have a similar contribution, which we shall denote
$\phi^{\mathrm L}_{-k,n}$ and which is not zero for $n >2k$;
however, it turns out that $\phi^{\mathrm L}_{-k,n} \to 0$
exponentially fast as $n\to \infty$.
\begin{lemma}
Let $k$ be a positive integer. One has
\begin{equation}
\label{lagnegiprop}
I_{\Lag}(-k,n) = 2^{-k}\sum_{j=0}^{n/2-1}\binom{2k+2j-1}{2j}\frac{(\frac{1}{2}(b+n))_{(-k-j)}}{(\frac{1}{2}(1+n))_{-(j)}} + \phi_{-k,n}^{\mathrm{L}},
\end{equation}
where 
\begin{equation}
\label{eq:phi_lag_neg}
\phi_{-k,n}^{\mathrm{L}} = \frac{\Gamma(n/2+b/2-1/2)}{\Gamma(n/2)\Gamma(b+n-1)}\sum_{j=0}^{k-1}\frac{\Gamma(k+j+n)\Gamma(b-k-j)2^{j}}{\Gamma(b/2+1/2-j)\Gamma(k+j+1)}.
\end{equation}
Furthermore, we have
\begin{equation}
 \label{eq:deay_rate}
 \phi_{-k,n}^{\mathrm{L}} = O\left(e^{-cn}\right), \quad n
 \to \infty, \quad c>0.
\end{equation}
\end{lemma}
\begin{proof}
The incomplete integral now becomes
\begin{equation}
\label{incompletenegative}
I_{\Lag}(-k,n) = \eta_{n/2-2}^{(1)}\gamma_{n-2}\int_{0}^{\infty}x^{\frac{b-1}{2}-k}e^{-x/2}\mathcal{L}_{n-1}^{b}(x)\epsilon\left[t^{\frac{b-1}{2}}e^{-t/2}\right](x)dx.
\end{equation}
As for positive moments, we insert into~\eqref{incompletenegative} the
identity 
\begin{equation}
\label{lastmin2}
\epsilon\left[t^{\frac{b-1}{2}}e^{-t/2}\right](x) = 
-\sum_{j=0}^{k-1}d_{j}^{(b+1)/2}x^{\frac{b-1}{2}-j}e^{-x/2}+\frac{1}{2}d_{k}^{(b+1)/2}
\epsilon\left[t^{\frac{b-1}{2}-k}e^{-t/2}\right](x),
\end{equation}
where $d_{j}^{b} =
2^{j+1}\frac{\Gamma(b)}{\Gamma(b-j)}$. Formula~\eqref{lastmin2} is
obtained in the same way as
equation~\eqref{incompletegammaidentity}. This gives two integrals:
the first one is
\begin{equation}
\begin{split}
\label{zerolagneg}
\phi_{-k,n}^{\mathrm{L}} & =
-\eta_{n/2-2}^{(1)}\gamma_{n-2}\sum_{j=0}^{k-1}\int_{0}^{\infty}d_{j}^{(b+1)/2}x^{b-j-k-1}e^{-x}\mathcal{L}_{n-1}^{b}(x)dx\\
&=-\eta_{n/2-2}^{(1)}\gamma_{n-2}\sum_{j=0}^{k-1}\sum_{i=0}^{n-1}\int_{0}^{\infty}\mathcal{C}_{i}^{-k-j-1,n-1}d_{j}^{(b+1)/2}x^{b-j-k-1}e^{-x}\mathcal{L}_{i}^{b-j-k-1}(x)dx,
\end{split}
\end{equation}
where we inserted the connection coefficients $\mathcal{C}_{i}$ for
the Laguerre polynomials (\ref{laguerreconnection}). Application of
orthogonality implies that the only contribution to the inner sum
occurs at $i=0$, yielding equation~\eqref{eq:phi_lag_neg}.  The
remaining non-trivial integral is
\begin{equation}
 \frac{1}{2}\eta_{n/2-2}^{(1)}\gamma_{n-2}d_{k}^{(b+1)/2}
\int_{0}^{\infty}x^{\frac{b-1}{2}-k}e^{-x/2}\mathcal{L}_{n-1}^{b}(x)\epsilon\left[t^{\frac{b-1}{2}-k}e^{-t/2}\right](x)dx,
\end{equation}
which can be computed in the same way as the right-hand side of
equation~\eqref{laguerrefinalline}. 
\end{proof}

\begin{remark}
 The sum  
\begin{equation}
   \label{eq:ord_orth}
2^{-k} \sum_{j=0}^{n/2-1}\binom{2k+2j-1}{2j}\frac{(n/2+b/2)_{(-k-j)}}%
 {(n/2+1/2)_{(-j)}}     
 \end{equation}
 in equation~\eqref{lagnegiprop} is $O(n^{-k})$ and of subleading
 order compared to $M^{(2)}_{\mathrm{L}_b}(-k,n-1)$, which gives the
 main contribution to the moments of the proper delay times for
 $\beta=1$.  However, it goes to zero much more slowly than the
 correction term $\phi^{\mathrm{L}}_{-k,n}$.
\end{remark}

\subsection{\label{sse:goe}Gaussian Orthogonal Ensemble}

The treatment of the GOE by our method is slightly different from the
LOE and JOE. We do not find a duality relation similar to
lemma~\ref{pr:duality}. However, the following proposition is
the analogue of equation (\ref{hersprop}) for the GSE.
\begin{lemma}
\label{le:ogf}
 Let $k$ be a positive integer and suppose $n$ is even. Then, the
 integral~\eqref{ointegrals} is explicitly given by
\begin{equation}
\label{gaussoprop}
O_{\mathrm{G}}(2k,n) = \sum_{j=1}^{\min(n/2-1,k)}
\sum_{i=0}^{\min(n/2-j-1,k-j)}\binom{k}{i}\binom{k}{i+j}
\frac{(n/2-i-j)_{(k+1/2)}}{(n/2-j)_{(1/2)}}.
\end{equation}
\end{lemma}
\begin{proof}
By changing variable of integration and  using
the relation~\eqref{hermitelaguerreconnection} we obtain
\begin{equation}
\label{ogaussdef}
\begin{split}
O_{\mathrm{G}}(2k,n) & = \int_{-\infty}^{\infty}\sum_{j=0}^{n/2-2}x^{k}
\gamma_{n-2}e^{-x^{2}}e_{j,n/2-2}^{(1)}\mathcal{H}_{n-1}(x)\mathcal{H}_{2j+1}(x)dx\\
&=\int_{0}^{\infty}\sum_{j=0}^{n/2-2}x^{k+1/2}\gamma_{n-2}e^{-x}e_{j,n/2-2}^{(1)}\mathcal{L}_{n/2-1}^{1/2}(x)\mathcal{L}_{j}^{1/2}(x)dx.
\end{split}
\end{equation}
Inserting the connection formula (\ref{laguerreconnection})
into (\ref{ogaussdef}) leads to
\begin{equation}
\begin{split}
O_{\mathrm{G}}(2k,n) & =
\gamma_{n-2}\sum_{j=0}^{n/2-2}e_{j,n/2-2}^{(1)}\sum_{i=0}^{j}\mathcal{C}_{i}^{k,j} \\
& \quad \times \sum_{p=0}^{n/2-1}\mathcal{C}_{p}^{k,n/2-1}
\int_{0}^{\infty}x^{k+1/2}e^{-x}\mathcal{L}_{n/2-1-p}^{k+1/2}(x)\mathcal{L}_{j-i}^{k+1/2}(x)dx.
\end{split}
\end{equation}
By using the orthogonality of the Laguerre polynomials and rearranging
the indices we obtain
\begin{equation}
O_{\mathrm{G}}(2k,n) = \gamma_{n-2}\sum_{j=1}^{n/2-1}e_{n/2-j-1,n/2-2}^{(1)}\sum_{i=0}^{n/2-j-1}\mathcal{C}_{i}^{k,n/2-j-1}\mathcal{C}_{i+j}^{k,n/2-1}h^{k+1/2}
_{n/2-j-i-1}.
\end{equation}
Inserting the appropriate constants from appendix~A completes the
proof.  The coefficients $\mathcal{C}_{i}^{k,n/2-j-1}$,
$\mathcal{C}_{i+j}^{k,n/2-1}$ and $h^{k+1/2} _{n/2-j-i-1}$ are those
for the Laguerre polynomials; the constants $\gamma_{n-2}$ and
$e_{n/2-j-1,n/2-2}^{(1)}$ are those associated to the GOE.
\end{proof}

The remaining task is the evaluation of the integral $I_{\mathrm{G}}(k,n)$ in
(\ref{incompleteintegral}). As previously, we obtain slightly different
expressions depending on whether $n\leq2k$ or $n > 2k$. For the latter
inequality the formula simplifies considerably.
\begin{lemma}
\label{le:Igf}
Let $n$ be an even integer.  If $n \le 2k$ we have
\begin{equation*}
\begin{split}
 \phi^{\mathrm{G}}_{k,n} = I_{\mathrm{G}}(2k,n)=
 &2^{n/2-k}\frac{(2k)!}{\Gamma(n/2)}
\sum_{j=0}^{k-n/2}\sum_{i=0}^{n/2-1}\frac{\binom{n-1}{2i}
\frac{2^{-j-2i}(-1)^{j}}{(2j+2i+1)j!}}{(k-n/2-j)!}\\
& \quad +\frac{(2k)!}{\Gamma(n/2)}\sum_{j=0}^{n/2-1}
\sum_{i=0}^{j}\frac{(n/2-i-1)!\binom{n-1}{n-2i-1}}{(j-i)!(k-j)!4^{k-i}}.
\end{split}
\end{equation*}
When  $n >2k$ we obtain
\begin{equation}
\phi^{\mathrm{G}}_{k,n} = (2k)!\sum_{j=0}^{k}
\frac{(n/2+1/2-j)_{(j)}2^{3j-k}}{(2j)!(k-j)!}.
\end{equation}
\end{lemma}
\begin{proof}
We have
\begin{equation}
 \label{incompleteherdef}
 I_{\mathrm G}(2k,n) = C\int_{-\infty}^{\infty}x^{2k}e^{-x^{2}/2}\mathcal{H}_{n-1}(x)\int_{-\infty}^{\infty}e^{-t^{2}/2}\mathrm{sgn}(x-t)dtdx
\end{equation}
where $C = \eta^{(1)}_{n/2-2}\gamma_{n-2}/2 =
(2\sqrt{\pi}\Gamma(n/2))^{-1}$.  Now, consider the generating function
\begin{equation*}
 \mathcal{M}_{\mathrm G}(s) = C\int_{-\infty}^{\infty}e^{sx}e^{-x^{2}/2}\mathcal{H}_{n-1}(x)\int_{-\infty}^{\infty}e^{-t^{2}/2}\mathrm{sgn}(x-t)dtdx.
\end{equation*}
Completing the square in the exponent and changing variables leads to
\begin{equation}
\label{motivating}
 \mathcal{M}_{\mathrm{G}}(s) =C\sum_{j=0}^{n-1}e^{s^{2}/2}\int_{-\infty}^{\infty}\mathcal{H}_{j}(u)s^{n-1-j}\binom{n-1}{j}\int_{-\infty}^{\infty}e^{-(v+s)^{2}/2}\mathrm{sgn}(u-v)dvdu,
\end{equation}
where we have applied the connection formula~\cite{HT03}
\begin{equation*}
\mathcal{H}_{n-1}(u+s) = \sum_{j=0}^{n-1}\binom{n-1}{j}\mathcal{H}_{j}(u)s^{n-1-j}.
\end{equation*}

Equation (\ref{motivating}) motivates us to study the function
\begin{equation}
\label{effj}
f_{j}(s) = \int_{-\infty}^{\infty}e^{-u^{2}/2}\mathcal{H}_{j}(u)\int_{-\infty}^{\infty}e^{-(v+s)^{2}/2}\mathrm{sgn}(u-v)dvdu.
\end{equation}
In appendix~B we compute $f_j(s)$ in terms of the power series
\begin{equation}
 f_{2j}(s) = \sum_{p=j}^{\infty}s^{2p+1}a^{+}_{p,j}, \quad \text{and}
 \quad 
f_{2j+1}(s) = 2j!\sqrt{\pi}+\sum_{p=j+1}^{\infty}s^{2p}a^{-}_{p,j},
\end{equation}
where
\begin{equation}
 a^{+}_{p,j} =
 \frac{2^{1-2p}\sqrt{\pi}(-1)^{j-p}}{(2p+1)\Gamma(p-j+1)}, 
\quad \text{and} \quad a^{-}_{p,j}=\frac{2^{2-2p}
\sqrt{\pi}(-1)^{j-p}}{(2p)\Gamma(p-j)}.
\end{equation}
Thus, we can write the decomposition
\begin{equation}
\mathcal{M}_{\mathrm G}(s) =
C\bigl(\mathcal{M}_{\mathrm G}^{+}(s)+\mathcal{M}_{\mathrm G}^{-}(s)\bigr), 
\end{equation}
where
\begin{align}
 \mathcal{M}_{\mathrm G}^{+}(s) & =
 \sum_{j=0}^{n/2-1}e^{s^{2}/2}\binom{n-1}{2j}s^{n-2j-1}f_{2j}(s),\\
 \mathcal{M}_{\mathrm G}^{-}(s) & =
 \sum_{j=0}^{n/2-1}e^{s^{2}/2}\binom{n-1}{2j+1}s^{n-2j-2}f_{2j+1}(s).
\end{align}
Computing the Taylor expansions of these functions is a routine
(though tedious) exercise. Eventually, we obtain
\begin{equation}
C\frac{d^{2k}}{ds^{2k}}\mathcal{M}_{\mathrm{G}}^{+}(s)\bigg|_{s=0} 
= 2^{n/2-k}\frac{(2k)!}{\Gamma(n/2)}
\sum_{j=0}^{k-n/2}\sum_{i=0}^{n/2-1}\frac{\binom{n-1}{2i}
\frac{2^{-j-2i}(-1)^{j}}{(2j+2i+1)j!}}{(k-n/2-j)!},
\end{equation}
which vanishes if $n>2k$.  We also have
\begin{equation}
\label{eq:last_effort}
 C\frac{d^{2k}}{ds^{2k}}\mathcal{M}_{\mathrm{G}}^{-}(s)
\bigg|_{s=0} = \frac{(2k)!}{\Gamma(n/2)}\sum_{j=0}^{\min(n/2-1,k)}
\sum_{i=0}^{j}\frac{(n/2-i-1)!\binom{n-1}{n-2i-1}}{(j-i)!(k-j)!4^{k-i}},
\end{equation}
which combined with (\ref{gaussoprop}) gives the equation
(\ref{beta1gauss}). Under the assumption $n>2k$,
equation~\eqref{eq:last_effort} can be
simplified further, leading to
\begin{equation*}
\begin{split}
I_{\mathrm{G}}(2k,n) & =
\frac{(2k)!}{\Gamma(n/2)}\sum_{i=0}^{k}\sum_{j=i}^{k}\frac{4^{i-k}(n/2-i-1)!\binom{n-1}{n-2i-1}}{(k-j)!(j-i)!}\\
&=(2k)!\sum_{j=0}^{k}\frac{(n/2+1/2-j)_{(j)}2^{3j-k}}{(2j)!(k-j)!},
\end{split}
\end{equation*}
where in the first equality we interchanged the order of summation,
while in the second one we have used the formula
\begin{equation}
\sum_{j=p}^{k}\frac{1}{(k-j)!(j-p)!}=\frac{2^{k-p}}{(k-p)!}.
\end{equation}
\end{proof}

\appendix
\section{The Coefficients $h_{n}$, $e^{(1)}_{n}$
 and $e^{(4)}_{n}$}
 \renewcommand{\theequation}{A\arabic{equation}}
 \setcounter{equation}{0}

The orthogonality normalisations $h_j$ defined
in~\eqref{orthogonality}  for the classical
orthogonal polynomials are tabulated in any standard reference on
special functions~\cite{AS72,Sze39}.  However, in
this article we work with \textit{monic} polynomials. This is not the
standard normalization found in the literature.  Therefore, for the
reader's convenience, we report the coefficients $h_{n}$,
 $e^{(1)}_{n}$,  $e^{(4)}_{n}$ that we use throughout
this paper.

We have
\begin{equation}
\begin{cases}
h_{j}^{a,b} =
\frac{\Gamma(a+j+1)\Gamma(b+j+1)\Gamma(j+1)\Gamma(a+b+j+1)}%
{\Gamma(a+b+2j+1)\Gamma(a+b+2j+2)}, & \text{Jacobi,} \\
h_{j}^{b} = \Gamma(j+1)\Gamma(b+j+1), & \text{Laguerre,} \\
h_{j} = j!2^{-j}\sqrt{\pi}, & \text{Hermite.}
\end{cases}
\end{equation}

Given these normalisations, recall that
\begin{equation}
\label{ees2app}
c_{n} = h_{n+1}h_{n}\gamma_{n},
\end{equation}
where
\begin{equation}
\label{ees3app}
 h_{n}\gamma_{n} =
\begin{cases}1, & \text{Hermite,}\\
\frac{1}{2},& \text{Laguerre,} \\
\frac{1}{2}(2n+a+b+2), & \text{Jacobi.}\end{cases}
\end{equation}
In \S\ref{se:orth_symp_sym} for the orthogonal and symplectic
ensembles we introduced the quantities
\begin{equation}
\label{ees_ap}
e_{j,n}^{(1)} =
\frac{h_{2n+2}}{c_{2j+1}}\prod_{i=j+1}^{n}\frac{c_{2i}}{c_{2i+1}},
\quad e_{j,n}^{(4)} = \frac{h_{2n+1}}{c_{2n}}
\prod_{i=j}^{n-1}\frac{c_{2i+1}}{c_{2i}}, 
\quad \eta^{(1)}_{n} = \prod_{j=0}^{n}\frac{c_{2j}h_{2j+2}}{c_{2j+1}h_{2j}}.
\end{equation}
We have the following explicit formulae:\\

\noindent JOE and JSE,
\begin{subequations}
\begin{align}
\label{jacobiees}
e_{j,n}^{(1),a,b} &=
\frac{16^{n-j}2\Gamma(n+3/2)\Gamma(n+3/2+a/2)}{\Gamma(j+3/2)\Gamma(j+3/2+a/2)\Gamma(j+3/2+b/2)}\notag\\
&\quad \times \frac{\Gamma(n+3/2+a/2+b/2)\Gamma(4j+4+a+b)}{\Gamma(j+3/2+a/2+b/2)\Gamma(4n+5+a+b)}\\
e_{j,n}^{(4),a,b}&=\frac{16^{n-j}2\Gamma(n+1)\Gamma(n+a/2+1)}{\Gamma(j+1)\Gamma(j+1+a/2)}\notag\\
&\quad \times \frac{\Gamma(n+b/2+1)\Gamma(n+a/2+b/2+1)\Gamma(4j+a+b+2)}{\Gamma(j+1+b/2)\Gamma(j+a/2+b/2+1)\Gamma(4n+a+b+3)}\\
\eta^{(1)}_{n}&=\frac{\Gamma(a/2+b/2+3/2+n)\Gamma(b/2+3/2+n)\Gamma(a/2+3/2+n)}{2^{-4n-4-a-b}\pi\Gamma(a+b+4n+5)}
\notag \\
&\quad \times \frac{\Gamma(a/2+b/2+1)\Gamma(n+3/2)}{\Gamma(b/2+1/2)\Gamma(a/2+1/2)};
\end{align}
\end{subequations}
LOE and LSE,
\begin{subequations}
\begin{align}
\label{laguerreees}
e_{j,n}^{(1),b} &= \frac{4^{n-j}2\Gamma(n+3/2)\Gamma(n+3/2+b/2)}{\Gamma(j+3/2)\Gamma(j+3/2+b/2)}\\
e_{j,n}^{(4),b} &= \frac{4^{n-j}2\Gamma(n+1)\Gamma(n+b/2+1)}{\Gamma(j+1)\Gamma(j+1+b/2)}\\
\eta^{(1)}_{n} &= 4^{n+1}\frac{\Gamma(n+3/2)\Gamma(n+b/2+3/2)}{\sqrt{\pi}\Gamma(b/2+1/2)};
\end{align}
\end{subequations}
GOE and GSE,
\begin{equation}
e_{j,n}^{(1)} = \frac{\Gamma(n+3/2)}{\Gamma(j+3/2)}, \quad  e_{j,n}^{(4)} = n!/j! \quad \text{and} \quad \eta^{(1)}_{n} = \frac{\Gamma(n+3/2)}{\sqrt{\pi}}.
\end{equation}

\section{The Generating Function $f_{j}(s)$}
\renewcommand{\theequation}{B\arabic{equation}}
 \setcounter{equation}{0}

In the proof of lemma~\ref{le:Igf} we needed to study the generating
function
\begin{equation}
\label{effj_ap}
f_{j}(s) = \int_{-\infty}^{\infty}e^{-u^{2}/2}\mathcal{H}_{j}(u)\int_{-\infty}^{\infty}e^{-(v+s)^{2}/2}\mathrm{sgn}(u-v)dvdu.
\end{equation}

The signed integral in the equation (\ref{effj_ap}) is closely related to
the error function, for which the $k^{th}$ derivative can be expressed
in terms of the monic Hermite polynomial
$\mathcal{H}_{k-1}(u)$. Differentiating under the integral, we find
\begin{equation}
\frac{1}{(2k)!} \frac{d^{2k}}{ds^{2k}}f_{2j}(s)\bigg|_{s=0} = \frac{-2^{k}}{(2k)!}\sqrt{2}\int_{-\infty}^{\infty}e^{-u^{2}}\mathcal{H}_{2j}(u)\mathcal{H}_{2k-1}(u/\sqrt{2})du=0,
\end{equation}
which follows from the oddness of the integrand. 

For the odd derivatives we get
\begin{subequations}
\begin{align}
\frac{1}{(2k+1)!}\frac{d^{2k+1}}{ds^{2k+1}}f_{2j}(s)\bigg|_{s=0} &= \frac{2^{k+1}}{(2k+1)!}\int_{-\infty}^{\infty}e^{-u^{2}}\mathcal{H}_{2j}(u)\mathcal{H}_{2k}(u/\sqrt{2})du \notag\\
&=\frac{2^{k+1}}{(2k+1)!}\int_{0}^{\infty}e^{-u}u^{-1/2}\mathcal{L}^{-1/2}_{j}(u)\mathcal{L}^{-1/2}_{k}(u/2)du \label{GOEuseconnectionformula}\\
\label{missing}
&=\sum_{i=0}^{k}\binom{k-1/2}{k-i}\frac{k!}{i!}\frac{2}{(2k+1)!}\int_{0}^{\infty}e^{-u}u^{-1/2}\mathcal{L}^{-1/2}_{j}(u)\mathcal{L}^{-1/2}_{i}(u)du\\
&=\frac{2(-1)^{k-j}\sqrt{\pi}4^{-k}}{(2k+1)(k-j)!}. \label{GOEuseorthogonality}
\end{align}
\end{subequations}
The expression in line (\ref{GOEuseconnectionformula}) was
obtained by means of the relations (\ref{hermitelaguerreconnection}), while
(\ref{missing}) follows from applying the connection
formula
\begin{equation}
\mathcal{L}_{k}^{b}(x/2) = 2^{-k}\sum_{i=0}^{k}\binom{b+k}{k-i}\frac{k!}{i!}(-1)^{i-k}\mathcal{L}_{i}^{b}(x).
\end{equation}
The last line (\ref{GOEuseorthogonality}) then follows from the
orthogonality (\ref{orthogonality}). Thus, we have
\begin{equation}
 f_{2j}(s) = \sum_{p=j}^{\infty}s^{2p+1}
\frac{2^{1-2p}\sqrt{\pi}(-1)^{j-p}}{(2p+1)\Gamma(p-j+1)}.
\end{equation}

In a similar fashion we find
\begin{equation}
f_{2j+1}(s) = 2j!\sqrt{\pi}+\sum_{p=j+1}^{\infty}s^{2p}\frac{2^{2-2p}\sqrt{\pi}(-1)^{j-p}}{(2p)\Gamma(p-j)},
\end{equation}
which required the double integral
\begin{equation}
\int_{-\infty}^{\infty}e^{-u^{2}/2}\mathcal{H}_{2j+1}(u)\int_{-\infty}^{\infty}e^{-v^{2}/2}\mathrm{sgn}(u-v)dvdu = 2j!\sqrt{\pi}.
\end{equation}

\bibliographystyle{amsplain}

\bibliography{chaotic_cav_bib}

\end{document}